\documentclass{amsart}

\usepackage{amsmath,amssymb}

\usepackage{cite}
\usepackage{mathdots}

\def\FF{\mathbb F}
\def\RR{\mathbb R}
\def\CC{\mathbb C}
\def\tr{\mathop{\top}\nolimits}
\def\Tr{\mathop{\mathrm {Tr}}\nolimits}
\def\diag{\mathop{\mathrm {diag}}\nolimits}

\def\rank{\mathop{\mathrm{rank}}\nolimits}

\newtheorem{theorem}{Theorem}
\newtheorem{lemma}{Lemma}
\newtheorem{corollary}{Corollary}
\newtheorem{proposition}{Proposition}

\theoremstyle{remark}
\newtheorem{remark}{Remark}
\newtheorem{example}{Example}

\newcommand{\cal}{\mathcal}

\newcommand{\inv}[1]{{\overline{#1}}}
\newcommand{\cqfd}{\hfill $\square$}

\newcommand\myeq{\mathrel{\overset{\makebox[0pt]{\mbox{\normalfont\tiny\sffamily def}}}{\Longleftrightarrow}}}

\newcounter{myenumi}

\newsavebox{\cmm}
\savebox{\cmm}{\indent}
\newenvironment{myenumerate}[1]{
\begin{list}{
{\bf #1~\themyenumi}. } {\labelwidth=0pt
\labelsep=0pt\leftmargin=0pt\usecounter{myenumi}} }{\end{list}}

\newenvironment{bew}[2]{\removelastskip\vspace{6pt}\noindent
 {\it Proof  #1.}~\rm#2}{\par\vspace{6pt}}

\begin{document}

\title{On Minkowski space and finite geometry}

\author{Marko Orel}
\address{University of Primorska, FAMNIT, Glagolja\v{s}ka 8, 6000 Koper, Slovenia}
\address{Institute of Mathematics, Physics and Mechanics, Jadranska 19, 1000 Ljubljana, Slovenia}
\address{University of Primorska, IAM, Muzejski trg 2, 6000 Koper, Slovenia}

\email{marko.orel@upr.si}

\begin{abstract}
The main aim of this interdisciplinary paper is to characterize all maps on finite Minkowski space of arbitrary dimension $n$ that map pairs of distinct light-like events into pairs of distinct light-like events. Neither bijectivity of maps nor preservation of light-likeness in the opposite direction, i.e. from codomain to domain, is assumed. We succeed in in many cases, which include the one with $n$ divisible by 4 and the one with $n$ odd and $\geq 9$, by showing that both bijectivity of maps and the preservation of light-likeness in the opposite direction is obtained automatically. In general, the problem of whether there exist non-bijective mappings that map pairs od distinct light-like events into pairs of distinct light-like events is shown to be related to one of the central problems in finite geometry, namely to existence of ovoids in orthogonal polar space. This problem is still unsolved in general despite a huge amount of research done in this area in the last few decades.

The proofs are based on the study of a core of an affine polar graph, which yields results that are closely related to the ones obtained previously by Cameron and Kazanidis for the point graph of a polar space.
\end{abstract}

\maketitle

\section{Introduction}

When Einstein introduced special relativity~\cite{Einstein} and derived the Lorentz transformation that relates the coordinates between two frames of reference, he assumed that this transformation is an affine map, that is, a sum of an additive and a constant map. It was later proved by Aleksandrov~\cite{aleksandrov1950}, that this assumption is redundant and that the constancy of speed of light, which is equal in both frames of reference, is sufficient. More precisely, he proved that any bijective map $\Phi$ on 4-dimensional Minkowski space-time, for which a pair of events are light-like if and only if their $\Phi$-images are light-like, is up to translation and dilation a Lorentz transformation, so an affine map. Maps characterized by Aleksandrov are sometimes referred as bijective maps that \emph{preserve the speed of light} in both directions (cf.~\cite{lester}). A very similar characterization was independently obtained by Zeeman~\cite{zeeman}, who assumed in addition that causality is preserved. A generalization of Aleksandrov theorem for $n$-dimensional ($n\geq 3$) Minkowski spaces was proved by himself~\cite{aleksandrov1967} and by Borchers and Hegerfeldt~\cite{borchers}. Hua observed that pairs of distinct events $(ct_1,x_1,y_1,z_1)$ and $(ct_2,x_2,y_2,z_2)$ in 4-dimensional Minkowski space-time are light-like if and only if their corresponding complex hermitian matrices $\big[\begin{smallmatrix}ct_1+x_1& y_1+i z_1\\
y_1-iz_1& ct_1-x_1\end{smallmatrix}\big]$ and $\big[\begin{smallmatrix}ct_2+x_2& y_2+i z_2\\
y_2-iz_2& ct_2-x_2\end{smallmatrix}\big]$ are \emph{adjacent}, that is, their difference is of rank one~\cite[Chapter~5]{hua_knjiga}. Bijective maps that preserve adjacency in both directions are characterized by fundamental theorems of geometry of matrices (see the book~\cite{wan}), so Hua was able to use such a theorem for $2\times 2$ complex hermitian matrices to reprove (according to~\cite[p.~97]{hua_knjiga} independently) Aleksandrov's and Zeeman's result on 4-dimensional Minkowski space-time. Recently, \v{S}emrl and Huang generalized the fundamental theorem for complex hermitian matrices~\cite{huang_semrl}. They did not assume that the maps that preserve adjacency are bijective, and the adjacency was assumed to be preserved in just one directions. In the language of Minkowski space-time, their result applied on $2\times 2$ matrices characterizes mappings (not assumed to be bijective) that map pairs of distinct light-like events into pairs of distinct light-like events (see~\cite[Section~3]{huang_semrl} and Theorem~4.2 in the arXiv version of~\cite{drugi_del}). It turns out that beside the maps that are already characterized by Aleksandrov's theorem (i.e., bijective maps with inverses that preserve the speed of light as well), the only other type of maps are those with the image contained in a set of pairwise light-like events. The existence of such maps is trivial since the field $\RR$ of real numbers is infinite.

Perhaps the first to consider an approximation between structures over the field of real numbers and structures over a very large finite field were astronomers Kustaanheimo~\cite{kustaanheimo1950,kustaanheimo1952} and J\"{a}rnefelt~\cite{jarnefelt1951}. Motivated by some alternative theories in particle physics, Lorentz transformations over finite fields were considered by many authors: see for example Coish~\cite{coish}, Shapiro~\cite{shapiro}, Ahmavaara~\cite{ahmavaara1,ahmavaara2,ahmavaara3,ahmavaara4}, Yahya~\cite{yahya}, Joos~\cite{joos}, Beltrametti and Blasi~\cite{belbla1968JMP,belbla1968NC}. Recently, Foldes~\cite{foldes} derived an approximation result between Lorentz transformations over real numbers and Lorentz transformations over a finite field. Blasi et al~\cite[Proposition~2]{blasi} obtained an analog of Aleksandrov's result for 4-dimensional Minkowski space over finite fields of prime cardinality. Lester~\cite{lester1977} generalized this theorem to more general spaces that include $n$-dimensional ($n\geq 3$) Minkowski spaces over finite fields. Recently, the author characterized maps on  4-dimensional Minkowski space over a finite field, that map pairs of distinct light-like events into pairs of distinct light-like events (see~\cite[Theorem~4.9]{drugi_del} and~\cite{FFA}). No bijectivity was assumed
and the `speed of light' was assumed to be preserved in just one directions.
However, it turned out that both bijectivity of maps and the preservation of the `speed of light' in the opposite direction is obtained automatically. In particular, there is a fundamental difference with the real case mentioned above, since there do not exist maps with the image contained in a set of pairwise light-like events. Nonexistence of such maps is not trivial, and several tools from graph theory were applied to prove it.

It is the aim of this paper to obtain an analog of~\cite[Theorem~4.9]{drugi_del} on $n$-dimensional case, that is, to  characterize maps on $n$-dimensional Minkowski space over a finite field that map pairs of distinct light-like events into pairs of distinct light-like events. It should be immediately emphasized that unlike the usual case, when generalizing certain problems from small number of dimensions to higher dimensional objects requires just some minor modifications in the proofs, here the situation is quite different. Firstly, and less importantly, the technique used in the proof of~\cite[Theorem~4.9]{drugi_del} translated the problem of characterizing maps that preserve the speed of light on 4-dimensional Minkowski space to the problem of characterizing maps that preserve adjacency on $2\times 2$ hermitian matrices over a finite field~\cite{FFA}. There is no such relation between higher dimensional Minkowski spaces and higher dimensional hermitian matrices, so this kind of technique will not work for our purposes in this paper. Secondly, and more importantly, we shall see that, in general, the characterization of maps that preserve the speed of light on $n$-dimensional Minkowski space can be different as the one obtained in 4 dimensions. Nevertheless, a completely analogous result as in 4 dimensions will be obtained for $n\equiv 0\ (\textrm{mod}~4)$, for odd $n\geq 9$, and for some other special choices of space dimension and field cardinality. Surprisingly, we will see that for the rest of the cases the question whether there exist nonbijective maps that map pairs of distinct light-like events into pairs of distinct light-like events (and these maps include those with the image contained in a set of pairwise light-like events) is related to a well known problem from finite geometry, namely to existence of \emph{ovoids} in orthogonal polar spaces. This problem remains unsolved in general, despite an extensive research in this area performed in the last few decades. Perhaps the quickest way, for a nonspecialist, to realize the amount of research done that is related to the existence of ovoids, is to check the MathSciNet database~\cite{mathscinet}, which currently contains almost 150 articles with the word `ovoid' or `ovoids' in the title (excluding those, where `ovoid' means something else).

The paper is organized as follows. In Section~\ref{grafi} and Section~\ref{matrike} we recall several results from graph and matrix theory, respectively, which are used in the sequel. Section~\ref{polarni}, where we investigate the core of an affine polar graph, is the fundamental part of the paper. It is divided in three subsections. In the first we prove that an affine polar graph is either a core or its core is a complete graph. The other two subsections explore  which of the two possibilities actually occur. Second subsection is based on spectral graph theory, while in the third subsection tools from finite geometry, which concern ovoids in orthogonal polar spaces, are used.  Results from Section~\ref{polarni} are applied in Section~\ref{sectionMinkowski}, where we classify mappings on finite Minkowski space that map pairs od distinct light-like events into pairs of distinct light-like events.

Beside the main results of the paper, that is, Theorem~\ref{glavni}, Theorem~\ref{noncomplete-core-par-hyper}, and Theorem~\ref{noncomplete-core-elliptic}, we prove several auxiliary results in the process that may be interesting on its own. In particular, Proposition~\ref{karakterji} (ii)  may be useful in graph theory, while Proposition~\ref{propThm2} together with Theorem~\ref{noncomplete-core-par-hyper} may be interesting from finite geometric perspective.

\section{Graph theory}\label{grafi}

All graphs in this paper are finite undirected with no loops and multiple edges. Such a \emph{graph} $\Gamma$ is a pair $(V,E)$ formed by the \emph{vertex set} $V=V(\Gamma)$, which is finite, and
the \emph{edge set} $E=E(\Gamma)$, which contains some unordered pairs of distinct vertices. The \emph{complement} $\inv{\Gamma}$ of graph $\Gamma$ is a graph with the same vertex set $V(\inv{\Gamma})=V(\Gamma)$ and edges defined by $$\{u,v\}\in E(\inv{\Gamma}), u\neq v\Longleftrightarrow \{u,v\}\notin E(\Gamma), u\neq v.$$
A graph $\Gamma'$ is a \emph{subgraph} of graph $\Gamma$ if $V(\Gamma')\subseteq V(\Gamma)$ and $E(\Gamma')\subseteq E(\Gamma)$. A subgraph $\Gamma'$ is \emph{induced} by the set $U\subseteq V(\Gamma)$ if $U=V(\Gamma')$ and $E(\Gamma')=\big\{\{u,v\}\in E(\Gamma)\ :\ u,v\in U\big\}$. Graph $\Gamma$ is \emph{connected} if for every pair of distinct vertices $u,v\in V(\Gamma)$, there is a path that connects them, i.e., there are vertices $u=v_0,v_1,\ldots,v_{t-1},v_t=v$ such that $\{v_{i-1},v_i\}\in E(\Gamma)$ for all $i$. The length of the shortest path between $u$ and $v$, that is, the minimal value $t$ is the \emph{distance} between $u$ and $v$, which is denoted by $d(u,v)$. The \emph{diameter} of a connected graph equals $\max _{u,v\in V(\Gamma)} d(u,v)$.
The \emph{neighborhood of vertex} $u\in V(\Gamma)$ is the subgraph that is induced by the set $\{v\in V(\Gamma)\ :\ d(u,v)=1\}$. The \emph{closed neighborhood of vertex} $u\in V(\Gamma)$ is the subgraph that is induced by the set $\{v\in V(\Gamma)\ :\ d(u,v)\leq 1\}$.

A \emph{graph homomorphism} between two graphs $\Gamma_1$ and
$\Gamma_2$ is a map $\Phi: V(\Gamma_1)\to V(\Gamma_2)$ such that the following implication holds for any
$v,u\in V(\Gamma_1)$:
\begin{equation}\label{i6}
\{u,v\}\in E(\Gamma_1) \Longrightarrow \{\Phi(u),\Phi(v)\}\in E(\Gamma_2).
\end{equation}
In particular, $\Phi(u)\neq \Phi(v)$ for any edge $\{u,v\}$. A bijective homomorphism for which  the
converse implication of~(\ref{i6}) also holds, is a \emph{graph isomorphism}. If $\Gamma_1=\Gamma_2$, a homomorphism is an \emph{endomorphism}  and an isomorphism is an \emph{automorphism}. Since the graphs considered are finite, an automorphism is the same as bijective endomorphism (cf.~\cite[Observation~2.3]{hahn_tardif}).
A graph $\Gamma$ is a \emph{core} if all its endomorphisms are automorphisms. If $\Gamma$ is a graph, then its subgraph $\Gamma'$ is
a \emph{core of $\Gamma$} if it is a core and there exists
some homomorphism $\Phi : \Gamma\to \Gamma'$. Any
graph has a core, which is always an induced subgraph and unique
up to isomorphism~\cite[Lemma~6.2.2]{godsil_knjiga}. If $\Gamma'$ is a core of $\Gamma$, then there exists a \emph{retraction} $\Psi : \Gamma\to \Gamma'$, that is, a homomorphism, which fixes the subgraph~$\Gamma'$. In fact, if $\Phi : \Gamma\to \Gamma'$ is any homomorphism, then the restriction $\Phi|_{\Gamma'}$ is an automorphism, since $\Gamma'$ is a core. Hence, the composition $\Psi:=(\Phi|_{\Gamma'})^{-1}\circ\Phi$ is a retraction onto~$\Gamma'$.
A graph is \emph{regular} of \emph{valency} $k$, if the neighborhood of arbitrary vertex consists of $k$ vertices.
A graph $\Gamma$ is \emph{vertex-transitive} if for every pair of vertices $u,v\in V(\Gamma)$ there exists some automorphism $\Phi$ of $\Gamma$ such that $\Phi(u)=v$. A vertex-transitive graph is always regular. An \emph{arc} of a graph is just a directed edge, that is, an ordered pair $(u,v)$ of adjacent vertices. A graph $\Gamma$ is \emph{arc-transitive} if for every pair of arcs $(u_1,v_1)$ and $(u_2,v_2)$ there exists some automorphism $\Phi$ of $\Gamma$ such that $(u_2,v_2)=\big(\Phi(u_1),\Phi(v_1)\big)$.
If a graph is vertex-transitive, then its core is vertex-transitive as well (cf.~\cite[Theorem~6.13.1]{godsil_knjiga}). Analogous result with the same proof holds for arc-transitivity~\cite[p.~128, Exercise~5]{godsil_knjiga}.

A \emph{complete graph} $K_t$ on $t$ vertices is a graph
such that every pair of distinct vertices form an edge. Complete graphs are cores.
A \emph{clique} in a graph $\Gamma$ is a subset $\mathcal{K}\subseteq V(\Gamma)$ that induces a complete subgraph.
A \emph{maximum clique} is a clique of the largest possible cardinality, which is referred as the \emph{clique number} $\omega(\Gamma)$ of graph $\Gamma$. A set $\mathcal{I}\subseteq V(\Gamma)$ is \emph{independent} if none of the pairs of vertices in $\mathcal{I}$ form an edge. The \emph{independence number} $\alpha(\Gamma)$ is the cardinality of the largest independent set in $\Gamma$. Since cliques and independent sets swap their roles in the complement of a graph, we have
\begin{equation}\label{e12}
\omega (\inv{\Gamma})=\alpha(\Gamma)\quad \textrm{and}\quad \alpha (\inv{\Gamma})=\omega(\Gamma).
\end{equation}
It is well known that, for a vertex-transitive graph,
\begin{equation}\label{e23}
\omega(\Gamma)\alpha(\Gamma)\leq |V(\Gamma)|,
\end{equation}
where $|V(\Gamma)|$ denotes the number of vertices in the graph (cf.~\cite[Corollary~3.11, p.~1471]{handcomb}). Moreover, the following holds.
\begin{lemma}\label{klikaneod}{\cite[p.~148]{cameron_kan}}
If the core of a vertex-transitive graph $\Gamma$ is a complete graph, then $\omega(\Gamma)\alpha(\Gamma)= |V(\Gamma)|$.
\end{lemma}

The \emph{lexicographic product} of graphs $\Gamma$ and $\Xi$
is the graph $\Gamma[\Xi]$
whose vertex set is $V(\Gamma)\times V(\Xi)$, and for which $\{(g_1, h_1), (g_2, h_2)\}$ is
an edge of $\Gamma[\Xi]$
precisely if $\{g_1, g_2\}\in E(\Gamma)$, or $g_1 = g_2$ and $\{h_1, h_2\}\in E(\Xi)$. It was proved in~\cite[Theorem~1]{stahl} that
\begin{equation}\label{e14}
\alpha\big(\Gamma[\Xi]\big)=\alpha(\Gamma)\alpha(\Xi).
\end{equation}
Since $\inv{\Gamma[\Xi]}=\inv{\Gamma}[\inv{\Xi}]$ (cf.~\cite[p.~57]{klavzar}), it follows from~(\ref{e12}) and~(\ref{e14}) that
\begin{equation}\label{e15}
\omega\big(\Gamma[\Xi]\big)=\omega(\Gamma)\omega(\Xi).
\end{equation}

Given a graph ƒ$\Gamma$, we use $\chi(\Gamma)$ to denote its chromatic number, that is, the smallest
integer $m$ for which there exists a vertex $m$-coloring, i.e., a map $\varphi : V(\Gamma)\to \{x_1, x_2,\ldots,x_m\}$ on the vertex set, which maps into a set of cardinality $m$ and satisfies $\varphi(u)\neq \varphi(v)$ whenever $\{u,v\}$ is an edge. Obviously, $\chi(\Gamma)\geq \omega(\Gamma)$. Equality holds if and only if the core of $\Gamma$ is a complete graph (cf.~\cite{godsil_clanek, cameron_kan}).

Let $G$ be a finite group and $S$ a subset of $G$ that is closed under inverses and does not contain the identity.
The \emph{Cayley graph} $\Gamma=\textrm{Cay}(G,S)$ is the graph with $G$ as its vertex set, two vertices $g$
and $h$ being joined by an edge if and only if $g^{-1}h\in S$. For any fixed $g\in G$, the map $h\mapsto g h$ is an automorphism of $\Gamma$, so Cayley graphs are vertex-transitive. The following result, which claims that a stronger statement than the one in Lemma~\ref{klikaneod} holds for a special type of Cayley graphs, was already proved in unpublished notes~\cite{godsil_notes} (see also the thesis~\cite{roberson_thesis}). We provide a proof for reader's convenience. Moreover, the proof contains a construction of a graph homomorphisms onto a maximum clique (\ref{coloring}), which will be relevant in Section~\ref{sectionMinkowski}.

\begin{proposition}\label{klikaneod2}
Let $\Gamma=\textrm{Cay}(G,S)$ be a Cayley graph such that the inverting map $g\mapsto g^{-1}$ is an automorphism of $\Gamma$. Then the following is equivalent:
\begin{enumerate}
\item a core of $\Gamma$ is a complete graph,
\item a core of $\inv{\Gamma}$ is a complete graph,
\item $\omega(\Gamma)\alpha(\Gamma)= |V(\Gamma)|$.
\end{enumerate}
\end{proposition}
\begin{remark}
It is easy to see that the map $g\mapsto g^{-1}$ is a graph automorphism if and only if the Cayley graph is \emph{normal} (cf.~\cite{roberson}). In particular, this is satisfied for graphs over abelian groups. However, there exist   Cayley graphs over non-abelian groups that are normal. An example of such is obtained if $G=GL_n(\FF_q)$ is the general linear group over a finite field $\FF_q$ with $q$ elements and $S\subseteq G$ is a subset of matrices with determinant $x$ or $x^{-1}$, where $x$ is a fixed generator of the cyclic multiplicative group $\FF_q\backslash\{0\}$.
\end{remark}
\begin{remark}\label{opomba}
There exist vertex-transitive graphs such that $(i)$ and $(ii)$ are not equivalent. An example of such is the graph $Q_{5}^{+}(5)$ (see Section~\ref{polarni} for the definition). In fact, it follows from~\cite[Theorem~3.5]{cameron_kan} and~\cite[Theorem~4.1]{penttila} that the core of $Q_{5}^{+}(5)$ is complete, while the core of $\overline{Q_{5}^{+}(5)}$ is not complete, as it follows from~\cite[Theorem~3.5]{cameron_kan} and~\cite[Theorem~6(b)]{thas1992}. In fact, $\overline{Q_{5}^{+}(5)}$ is a core.
\end{remark}
\begin{remark}
A graph is a \emph{CIS graph} if every its maximal clique and maximal independent set (according to inclusion) intersect. A characterization of CIS circulants was obtained in~\cite[Theorem~3]{martin1}. Subsequently, the result was generalized for arbitrary CIS vertex-transitive graphs~\cite{martin2}. In particular, it was proved that equation $(iii)$ from  Proposition~\ref{klikaneod2} is satisfied for any such graph. Hence, any CIS normal Cayley graph has a complete core.
\end{remark}

\begin{bew}{of Proposition~\ref{klikaneod2}}
Since (\ref{e12}) holds and a complement of a vertex-transitive graph is vertex-transitive, both $(i)$ and $(ii)$ implies $(iii)$ by Lemma~\ref{klikaneod}.

Assume now that $(iii)$ holds. Let $\mathcal{K}$  and  $\mathcal{I}$ be a clique of size $\omega(\Gamma)$ and an independent set of size $\alpha(\Gamma)$, respectively. We claim that
\begin{equation}\label{e13}
G=\mathcal{K}\cdot \mathcal{I}:=\{ki\ :\ k\in \mathcal{K}, i\in \mathcal{I}\}.
\end{equation}
By $(iii)$ it suffices to show that $k_1 i_1\neq k_2 i_2$ for $(k_1,i_1)\neq (k_2,i_2)$. So assume that $k_1 i_1= k_2 i_2$. Since  $\mathcal{I}$ is an independent set and the inverting map is a graph automorphism, it follows that $k_1^{-1}k_2=i_1 i_2^{-1} \notin S$. Since $\mathcal{K}$ is a clique, we deduce that $k_1=k_2$, and consequently $i_1=i_2$. Hence, (\ref{e13}) holds and any $g\in G$ can be uniquely written as a product $g=k_g i_g$, $k_g\in \mathcal{K}, i_g\in \mathcal{I}$.

We claim that the map $G\to \mathcal{K}$, defined by
\begin{equation}\label{coloring}
g\mapsto k_g,
\end{equation}
is a homomorphism between $\Gamma$ and a complete subgraph induced by $\mathcal{K}$, which is therefore a core of $\Gamma$, so $(i)$ holds. Since $\mathcal{K}$ is a clique, it suffices to show that adjacent vertices are not mapped in the same vertex. This certainly holds, since $k_{g_1}=k_{g_2}$ implies $g_1^{-1} g_2=i_{g_1}^{-1} i_{g_2}\notin S$, which means that $g_1$ and $g_2$ are not adjacent.

Similarly we see that the map $G\to \mathcal{I}$, defined by $g\mapsto i_g$, is a homomorphism between $\inv{\Gamma}$ and a complete subgraph induced by $\mathcal{I}$, which is therefore a core of $\inv{\Gamma}$, so $(ii)$ holds.\cqfd
\end{bew}

The \emph{spectrum} of graph $\Gamma$ consists of the eigenvalues (and their multiplicities) of its \emph{adjacency matrix}, i.e., the $|V(\Gamma)|\times |V(\Gamma)|$ real binary matrix with 1 at position $(i, j)$ if the $i$-th and $j$-th vertices are adjacent and 0 otherwise (order of vertices is arbitrary). The next lemma is the well known Hoffman upper bound for the independence number (cf.~\cite[Theorem~3.5.2]{brouwerhaemers}).
\begin{lemma}\label{hoffman}
If $\Gamma$ is a regular graph with nonzero valency, then
\begin{equation}\label{hoffman_neodvisna}
\alpha(\Gamma)\leq |V(\Gamma)|\cdot \frac{-\lambda_{\min}}{\lambda_{\max}-\lambda_{\min}},
\end{equation}
where $\lambda_{\max}$ and $\lambda_{\min}$ are the largest and smallest eigenvalue of $\Gamma$, respectively.
If an independent set ${\cal I}$ meet this bound, then every vertex not in ${\cal I}$ is adjacent to precisely $-\lambda_{\min}$ vertices of ${\cal I}$.
\end{lemma}

Let $G$ be a finite abelian group. A \emph{character} $\xi$ of $G$ is a group homomorphism $\xi: G\to\CC\backslash\{0\}$, that is, a nowhere-zero complex map that satisfies $\xi(gh)=\xi(g)\xi(h)$ for all $g,h\in G$. It is well known, that there exist precisely $|G|$ distinct characters (cf.~\cite[Theorem~5.5]{finite-fields-LN}). Moreover, they determine the spectrum of Cayley graphs on $G$.
\begin{lemma}\label{karakterji-spekter} (cf.~\cite[Subsection~1.4.9]{brouwerhaemers}).
If $G$ is a finite abelian group, then the eigenvalues of $Cay(G,S)$ are precisely the values $\sum_{s\in S}\xi(s)$, where $\xi$ ranges over all characters of $G$.
\end{lemma}

\section{Matrix theory}\label{matrike}

From now on, if not stated otherwise, $\FF_q$ is a finite field with $q$ elements, where $q=p^k$ is a power of an odd prime $p$. The number of nonzero squares and the number of non-squares in $\FF_q$ both equal $\frac{q-1}{2}$. More precisely, for any non-square $d\in\FF_q$, the field is of the form
\begin{equation}\label{e41}
\FF_q=\{0,x_1^2,\ldots, x_{\frac{q-1}{2}}^2, d x_1^2,\ldots, d x_{\frac{q-1}{2}}^2\}
\end{equation}
for some nonzero $x_1,\ldots,x_{\frac{q-1}{2}}\in\FF_q$ (cf.~\cite[Theorem~6.20]{wan2}). Let $n\geq 2$ be an integer. Let $GL_n(\FF_q)$ denote the set of all $n\times n$ invertible  matrices with coefficients in $\FF_q$. Similarly let $SGL_n(\FF_q)\subseteq GL_n(\FF_q)$ be the subset of all symmetric matrices, that is, the elements of $SGL_n(\FF_q)$ satisfy $A^{\tr}=A$ and $\det(A)\neq 0$, where $\tr$ and $\det$ are the transpose and the determinant of a matrix, respectively. The set of all column vectors of dimension $n$, i.e. $n\times 1$ matrices,  with entries from the field $\FF_q$, is denoted by $\FF_q^n$. Its elements are written in bold style, like~${\bf x}$.

We now recall some auxiliary results from matrix theory. The first lemma is well known and it calculates the determinant of a rank--one perturbation of an invertible matrix.
\begin{lemma}{(cf.~\cite[Chapter~14]{handbook})}\label{lema_tsats}
Let $A\in GL_n(\FF_q)$ and  ${\bf x},{\bf y}\in \FF_q^{n}$. Then
$$\det(A+{\bf x}{\bf y}^{\tr})=(\det A)\cdot (1+{\bf y}^{\tr}A^{-1}{\bf x}).$$
\end{lemma}

\begin{corollary}\label{matrika}
Let ${\bf a}:=(a_1,\ldots,a_n)^{\tr}\in\FF_q^n$, where $a_1\neq 0$. Then,
\begin{equation}\label{e27}
\det \left(\begin{array}{cccc}1+\frac{a_2^2}{a_1^2}&\frac{a_2 a_3}{a_1^2}&\ldots&\frac{a_2 a_n}{a_1^2}\\
\frac{a_2 a_3}{a_1^2}&1+\frac{a_3^2}{a_1^2}&\ldots&\frac{a_3 a_n}{a_1^2}\\
\vdots&\vdots&\ddots&\vdots\\
\frac{a_2 a_n}{a_1^2}&\frac{a_3 a_n}{a_1^2}&\ldots&1+\frac{a_n^2}{a_1^2}
\end{array}\right)=\frac{{\bf a}^{\tr}{\bf a}}{a_1^2}.
\end{equation}
\end{corollary}
\begin{proof}
Matrix in (\ref{e27}) equals $I_{n-1}+\frac{1}{a_1^2}{\bf b}{\bf b}^{\tr}$, where $I_{n-1}$ is the $(n-1)\times (n-1)$ identity matrix and ${\bf b}=(a_2,\ldots,a_n)^{\tr}$. The proof ends by Lemma~\ref{lema_tsats}.
\end{proof}

\begin{lemma}\label{lema33}
Let $A\in GL_{n-1}(\FF_q)$, ${\bf x}\in \FF_q^{n-1}$, and $a\in\FF_q$. Matrix
\begin{equation}\label{e31}
\left(\begin{array}{cl}
a&{\bf x}^{\tr}\\
{\bf x}&A
\end{array}\right)
\end{equation}
is singular if and only if $a={\bf x}^{\tr}A^{-1}{\bf x}$ in which case there is $P\in GL_n(\FF_q)$ such that
\begin{equation}\label{e32}
\left(\begin{array}{cl}
{\bf x}^{\tr}A^{-1}{\bf x}&{\bf x}^{\tr}\\
{\bf x}&A
\end{array}\right)=P^{\tr}\left(\begin{array}{cc}
0&0\\
0&A
\end{array}\right)P.
\end{equation}
\end{lemma}
\begin{proof}
If $a\neq {\bf x}^{\tr}A^{-1}{\bf x}$, then it is straightforward to check that matrix (\ref{e31}) has inverse
$$\left(\begin{array}{cc}
0&0\\
0&A^{-1}
\end{array}\right)+\frac{1}{a-{\bf x}^{\tr}A^{-1}{\bf x}}
\left(\begin{array}{c}
-1\\
A^{-1}{\bf x}
\end{array}\right)
\left(\begin{array}{c}
-1\\
A^{-1}{\bf x}
\end{array}\right)^{\tr}.$$ For $a={\bf x}^{\tr}A^{-1}{\bf x}$, the matrix (\ref{e31}) equals (\ref{e32}), where $$P=\left(\begin{array}{cc}
1&0\\
A^{-1}{\bf x}&I_{n-1}
\end{array}\right),$$
so it is singular.
\end{proof}
Lemma~\ref{lemmawan1993} can be found in~\cite[Theorem~6.8]{wan1993}. It can be also deduced from celebrated Witt's Theorem (cf.~\cite[Theorem~8, p.~167]{jacobson} or~\cite[Theorem~3.9]{artin}), by using a similar procedure as in~\cite[Lemma~2.3]{prvi_del}.
\begin{lemma}\label{lemmawan1993}
Let $A\in SGL_n(\FF_q)$ and let $Q_1, Q_2$ be two $n\times m$ matrices with coefficients in $\FF_q$ such that
$\rank Q_1=m=\rank Q_2$. Then there is $P\in GL_n(\FF_q)$ such that $P^{\tr}AP=A$ and $PQ_1=Q_2$ if and only if $Q_1^{\tr}AQ_1=Q_2^{\tr}AQ_2$.
\end{lemma}

\begin{corollary}\label{witt-p1}
Let $A\in SGL_n(\FF_q)$, ${\bf x}_1,{\bf x}_2\in\FF_q^n$, and assume $a_1,a_2,d\in\FF_q$ are nonzero, where $d$ is a non-square. If either ${\bf x}_i^{\tr}A{\bf x}_i=a_i^2$ for $i=1,2$ or ${\bf x}_i^{\tr}A{\bf x}_i=da_i^2$ for $i=1,2$, then there is $P\in GL_n(\FF_q)$ such that $$P {\bf x}_1=\frac{a_1}{a_2}{\bf x}_2\quad \textrm{and}\quad P^{\tr}AP=A.$$
\end{corollary}
\begin{proof}
Let $Q_1={\bf x}_1$ and $Q_2=\frac{a_1}{a_2}{\bf x}_2$. Then $Q_1^{\tr}AQ_1=Q_2^{\tr}AQ_2$, so the result follows from Lemma~\ref{lemmawan1993}.
\end{proof}

\begin{corollary}\label{lema4}
Let $A\in SGL_n(\FF_q)$. For $i=1,2$ assume that ${\bf x}_i,{\bf y}_i\in\FF_q^n$ satisfy ${\bf x}_i^{\tr}A{\bf x}_i=0={\bf y}_i^{\tr}A{\bf y}_i$ and ${\bf x}_i^{\tr}A{\bf y}_i\neq 0$. Then there exist $P\in GL_n(\FF_q)$ and  nonzero $\alpha\in\FF_q$ such that $P^{\tr}AP=A$,
$P{\bf x}_1={\bf x}_2$ and  $P{\bf y}_1=\alpha {\bf y}_2$.
\end{corollary}
\begin{proof}
Define $\alpha:=\frac{{\bf x}_1^{\tr}A{\bf y}_1}{{\bf x}_2^{\tr}A{\bf y}_2}$. Let $Q_1$ and $Q_2$ be $n\times 2$ matrices with ${\bf x}_1, {\bf y}_1$ and ${\bf x}_2, \alpha {\bf y}_2$ as their columns, respectively. It follows from the assumptions that $\rank Q_1=2=\rank Q_2$, so the proof ends by Lemma~\ref{lemmawan1993}.
\end{proof}

\section{Polar spaces}\label{polarni}

The main purpose of this section is to investigate what is a core of an affine polar graph $VO_n^{\varepsilon}(q)$ (see below for the definition). In the first subsection we start by recalling few more definitions and properties of finite fields and related structures. We end it with the result that shows that the core of $VO_n^{\varepsilon}(q)$ is either complete or the graphs itself is a core.

\subsection{Preliminaries and auxiliary results}

Recall that $q=p^k$ is a power of an odd prime $p$.  Let
$\FF_p:=\{0,1,2,\ldots,p-1\}\subseteq \FF_q$ be the prime subfield, that is,
$\FF_p=\{x\in \FF_q\,:\, x^p=x\}$.
The \emph{trace map} $\Tr : \FF_q\to \FF_p$,
defined by $\Tr(x):=x+x^p+\ldots+x^{p^{k-1}}$, is $\FF_p$-linear and surjective (cf.~\cite[Theorem~2.23]{finite-fields-LN}), while
\begin{equation}\label{sled}
\Tr(x)=0\Longleftrightarrow x=y^p-y
\end{equation}
for some $y\in\FF_q$ (cf.~\cite[Theorem~2.25]{finite-fields-LN}).
Moreover, any $\FF_p$-linear map $\psi:\FF_q\to \FF_p$ is of the form
\begin{equation}\label{sled-lin}
\psi(x)=\Tr(yx)
\end{equation}
for some $y\in\FF_q$ (cf.~\cite[Theorem~2.24]{finite-fields-LN}). Lemma~\ref{lema1} follows immediately from~\cite[Theorems~6.26 and~6.27]{finite-fields-LN}.
\begin{lemma}\label{lema1}
Let $b\in \FF_q$ and $A\in SGL_n(\FF_q)$. Then
$$\big|\{{\bf x}\in\FF_q^n : {\bf x}^{\tr}A{\bf x}=b\}\big|=\left\{ \begin{array}{ll} q^{n-1}+v(b)q^{\frac{n-2}{2}}\eta\big((-1)^{\frac{n}{2}}\det(A)\big) &  \ \textrm{if}\ n\ \textrm{is even},\\
q^{n-1}+q^{\frac{n-1}{2}}\eta\big((-1)^{\frac{n-1}{2}}b\det(A)\big) &  \
\textrm{if}\ n\ \textrm{is odd}.\end{array} \right.$$
Here, the integer valued maps
$v: \FF_q\to \mathbb{Z}$ and
$\eta: \FF_q\to \mathbb{Z}$ are defined as follows: $v(0):=q-1$, $v(b):=-1$ if $b\neq 0$,
$\eta(0):=0$, $\eta(d):=1$ if $d$
is a nonzero square, and $\eta(d)=-1$ if $d$ is a non-square in $\FF_q$.
\end{lemma}

Let $\langle {\bf x} \rangle$ be the one-dimensional vector subspace in $\FF_q^n$ that is spanned by nonzero column vector ${\bf x}\in\FF_q^n$. For $A\in SGL_n(\FF_q)$, the set $\{\langle {\bf x} \rangle : {\bf x}^{\tr}A{\bf x}=0, {\bf x}\neq 0\}$ is a \emph{quadric}. If $n$ is odd, then the quadric is \emph{parabolic} and has $\frac{q^{n-1}-1}{q-1}$ elements by Lemma~\ref{lema1}. If $n$ is even, then, by Lemma~\ref{lema1}, the quadric has either $\frac{(q^{n/2}-1)(q^{n/2-1}+1)}{q-1}$ or  $\frac{(q^{n/2}+1)(q^{n/2-1}-1)}{q-1}$ elements. In the first case the quadric is \emph{hyperbolic}, while in the second case it is \emph{elliptic}. Every quadric represents the vertex set of a \emph{point graph of an orthogonal polar space}, where two distinct elements $\langle {\bf x} \rangle$ and $\langle {\bf y} \rangle$ of a quadric form an edge if and only if
\begin{equation}\label{e2}
{\bf x}^{\tr} A {\bf y}=0.
\end{equation}
Two graphs constructed from two different matrices $A,A'\in SGL_n(\FF_q)$ are isomorphic, provided that the two corresponding quadrics are of the same type. The graphs obtained in this way are denoted by $Q_{n-1}(q)$, $Q_{n-1}^{+}(q)$, and $Q_{n-1}^{-}(q)$ if the quadric is parabolic, hyperbolic, and elliptic, respectively. In the general case, when it is not specified, which of the three types is meant, we write $Q_{n-1}^{\varepsilon}(q)$. Graph $Q_{n-1}^{\varepsilon}(q)$ is vertex-transitive (cf.~\cite{cameron_kan}).
Maximum cliques in a point graph of an orthogonal polar space are formed by \emph{maximal totally isotropic subspaces}, also referred as \emph{generators}, which consist of
\begin{equation}\label{e16}
s:=\frac{q^{r}-1}{q-1}
\end{equation}
quadric elements $\langle {\bf x}_1 \rangle,\ldots, \langle {\bf x}_s \rangle$ that satisfy
${\bf x}_j^{\tr} A {\bf x}_k=0$
for all $j,k$.
Here $r$ is the Witt index (cf.~\cite{cameron_kan,thas_hand}), that is,
\begin{equation}\label{wittindex}
r=\left\{ \begin{array}{ll}  \frac{n-1}{2} &  \ \textrm{in parabolic case},\\
\frac{n}{2}  &  \ \textrm{in hyperbolic case},\\
\frac{n}{2}-1  &  \ \textrm{in elliptic case},\end{array} \right.
\end{equation}
so the clique numbers equal
\begin{equation}\label{e1}
\omega\big(Q_{n-1}(q)\big)=\tfrac{q^{(n-1)/2}-1}{q-1},\ \omega\big(Q_{n-1}^{+}(q)\big)=\tfrac{q^{n/2}-1}{q-1},\ \omega\big(Q_{n-1}^{-}(q)\big)=\tfrac{q^{n/2-1}-1}{q-1}.
\end{equation}
In the literature devoted to polar spaces (cf.~\cite{thas_hand}), there is a usual assumption for Witt index to be at least two, since the case $r<2$ is not interesting. However, for the discussion in the sequel we prefer that $Q_{n-1}^{\varepsilon}(q)$ is defined for all $n\geq 2$. So, $Q_{1}^{-}(q)$ is an empty graph (i.e., a graph without vertices), while $Q_{1}^{+}(q)$, $Q_{2}(q)$, and $Q_{3}^{-}(q)$ are graphs with $2$, $q+1$, and $q^2+1$ isolated vertices (i.e., there are no edges), respectively.

\emph{Affine polar graphs} are defined similarly as point graphs of orthogonal polar spaces (cf.~\cite{www}).
Here, the vertex set equals $\FF_q^n$ ($n\geq 2$) and two distinct column vectors ${\bf x}$ and ${\bf y}$ form an edge if and only if
\begin{equation}\label{e3}
({\bf x}-{\bf y})^{\tr}A({\bf x}-{\bf y})=0,
\end{equation}
where $A\in SGL_n(\FF_q)$ is a fixed matrix. Two graphs constructed from two different matrices $A,A'\in SGL_n(\FF_q)$ are isomorphic, provided that the two quadrics defined by $A$ and $A'$ are of the same type. Graphs obtained in this way are denoted by $VO_n(q)$, $VO_n^{+}(q)$, and $VO_n^{-}(q)$, if the corresponding quadric is parabolic, hyperbolic, and elliptic, respectively. In the general case, when it is not specified, which of the three types is meant, we write $VO_n^{\varepsilon}(q)$. Observe that $VO_n^{\varepsilon}(q)$ is the Cayley graph $\textrm{Cay}(G,S)$ for the additive group $G:=(\FF_q^n,+)$ and the set
$S:=\{{\bf x}\in\FF_q^n\backslash\{0\} : {\bf x}^{\tr}A{\bf x}=0\}$, so it is a vertex-transitive graph. If ${\bf x}^{\tr}A{\bf x}=0={\bf y}^{\tr}A{\bf y}$, then equations (\ref{e2}) and (\ref{e3}) are equivalent. Hence, by identifying  $\langle {\bf x} \rangle$ with any nonzero multiple of ${\bf x}$, we deduce that $Q_{n-1}^{\varepsilon}(q)$ is (isomorphic to) a subgraph in $VO_n^{\varepsilon}(q)$.
Moreover, the following holds.

\begin{lemma}\label{lexi}
Let $N$ be the neighborhood of any vertex in $VO_n^{\varepsilon}(q)$. Then $N$ is isomorphic to lexicographic product $Q_{n-1}^{\varepsilon}(q)[K_{q-1}]$.
\end{lemma}
\begin{proof}
If we consider the graph $VO_2^{-}(q)$, then $N$ and $Q_{1}^{-}(q)[K_{q-1}]$ are empty graphs, hence isomorphic. Assume now that $VO_n^{\varepsilon}(q)\neq VO_2^{-}(q)$. Since $VO_n^{\varepsilon}(q)$ is a vertex-transitive graph, we may assume that $N$ is the neighborhood of the zero vertex, that is, the subgraph induced by the set $\{{\bf x}\in\FF_q^n\backslash\{0\} : {\bf x}^{\tr}A{\bf x}=0\}$. The complete graph $K_{q-1}$ can be viewed as the graph on the vertex set $\FF_q\backslash\{0\}$, where distinct scalars are adjacent. Fix nonzero ${\bf x}_1,\ldots, {\bf x}_t\in \FF_q^n$ such that $\langle{\bf x}_1\rangle,\ldots, \langle{\bf x}_t\rangle$ are precisely all distinct vertices of $Q_{n-1}^{\varepsilon}(q)$. Then, the map $(\langle{\bf x}_i\rangle, a)\mapsto a {\bf x}_i$, with $a\in \FF_q\backslash\{0\}$, is the desired isomorphism from $Q_{n-1}^{\varepsilon}(q)[K_{q-1}]$ onto $N$.
\end{proof}

Corollary~\ref{lema2} is well known. It can be easily deduced also from Lemma~\ref{lexi}.

\begin{corollary}\label{lema2}
The clique numbers of affine polar graphs equal
\begin{equation}\label{eq-klike-affine}
\omega\big(VO_n(q)\big)=q^{(n-1)/2},\quad \omega\big(VO_n^{+}(q)\big)=q^{n/2},\quad \omega\big(VO_n^{-}(q)\big)=q^{n/2-1}.
\end{equation}
\end{corollary}
\begin{proof}
Let $N$ be the neighborhood of the zero vector in $VO_n^{\varepsilon}(q)$,
$\mathcal{K}$ a maximum clique in  $VO_n^{\varepsilon}(q)$, and ${\bf x}_0\in \mathcal{K}$ arbitrary. Then $\mathcal{K}-{\bf x}_0:=\{{\bf x}-{\bf x}_0 : {\bf x}\in\mathcal{K}\}$ is a maximum clique, which is contained in the closed neighborhood $N_0$ of the zero vector. Hence,
$$\omega\big(VO_n^{\varepsilon}(q)\big)=|\mathcal{K}|=|\mathcal{K}-{\bf x}_0|=\omega\big(N_0\big)=\omega(N)+1.$$
The proof ends by Lemma~\ref{lexi}, (\ref{e15}), and (\ref{e1}).
\end{proof}
Note that any clique of size (\ref{eq-klike-affine}) is a (totaly isotropic) vector space or its translation.

\begin{lemma}\label{prop1}
Graph $VO_n^{\varepsilon}(q)$ is arc-transitive.
\end{lemma}
\begin{proof}
Let $({\bf x}_1,{\bf y}_1)$ and $({\bf x}_2,{\bf y}_2)$ be two arcs, that is, $({\bf y}_i-{\bf x}_i)^{\tr}A({\bf y}_i-{\bf x}_i)=0$ and ${\bf y}_i\neq {\bf x}_i$ for $i=1,2$, where $A$ is the defining matrix (\ref{e3}). If we evaluate Lemma~\ref{lemmawan1993} at $Q_i:={\bf y}_i-{\bf x}_i$, we obtain $P\in GL_n(\FF_q)$ such that
$P^{\tr}AP=A$ and
\begin{equation}\label{e42}
P({\bf y}_1-{\bf x}_1)={\bf y}_2-{\bf x}_2.
\end{equation}
Then all three maps
$$\Phi_1({\bf z}):={\bf z}-{\bf x}_1,\qquad \Phi_2({\bf z}):=P{\bf z},\qquad \Phi_3({\bf z}):={\bf z}+{\bf x}_2$$
are automorphisms of $VO_n^{\varepsilon}(q)$, so the same holds for their composition
$\Phi:=\Phi_3\circ \Phi_2\circ \Phi_1$, which is given by $\Phi({\bf z})=P{\bf z}-(P{\bf x}_1-{\bf x}_2)$. By (\ref{e42}) it satisfies $\Phi({\bf x}_1)={\bf x}_2$ and $\Phi({\bf y}_1)={\bf y}_2$, so the graph is arc-transitive.
\end{proof}

It is well known that graphs  $VO_n^{-}(q)$ and $VO_n^{+}(q)$ are strongly regular~\cite[$C.12.^\pm$]{hubaut} (see~\cite{BCN} for the definition of a strongly regular graph). So, except for the graph $VO_2^{-}(q)$, which is formed by $q^2$ isolated vertices, all other graphs are connected with diameter 2 (cf.~\cite[p.~4]{BCN}). The same holds for graph $VO_n(q)$ as shown below.

\begin{proposition}\label{diameter}
Let $n\geq 3$ be odd. Graph $VO_n(q)$ is connected with diameter 2.
\end{proposition}
\begin{proof}
Since various defining invertible matrices $A$ in~(\ref{e3}) produce isomorphic graphs, we may assume that
$A=\diag(1-d,d,-1,\ldots,-1)$ is a diagonal matrix, where $d\in\FF_q$ is a non-square. Obviously, the diameter is at least two. Let ${\bf x},{\bf y}\in\FF_q^n$ be two non-adjacent vertices, that is, $({\bf x}-{\bf y})^{\tr}A({\bf x}-{\bf y})\neq 0$. By (\ref{e41}), there is nonzero  $a_1\in \FF_q$ such that $({\bf x}-{\bf y})^{\tr}A({\bf x}-{\bf y})\in\{da_1^2, a_1^2\}$. To deduce that the graph is connected and the diameter equals two, we need to find ${\bf z}\in\FF_q^n$ that is adjacent to both ${\bf x}$ and ${\bf y}$. We separate two cases.
\begin{myenumerate}{Case}
\item Let $({\bf x}-{\bf y})^{\tr}A({\bf x}-{\bf y})=da_1^2$. Let ${\bf e}_i=(0,\ldots,0,1,0,\ldots,0)^{\tr}$ be the $i$-th standard vector. Define ${\bf w}:={\bf e}_2$ and ${\bf u}:=\frac{1}{2}({\bf e}_1+{\bf e}_2+{\bf e}_3)$. Since ${\bf w}^{\tr}A{\bf w}=d=d\cdot 1^2$, Corollary~\ref{witt-p1} shows that there is
    $P\in GL_n(\FF_q)$ such that $P ({\bf y}-{\bf x})=a_1{\bf w}$ and $P^{\tr}AP=A$. Define ${\bf z}:={\bf x}+a_1 P^{-1} {\bf u}$. Since $A=(P^{-1})^{\tr}AP^{-1}$ and ${\bf y}-{\bf x}=a_1P^{-1}{\bf w}$, we easily see that
    $({\bf z}-{\bf x})^{\tr}A({\bf z}-{\bf x})=0=({\bf y}-{\bf z})^{\tr}A({\bf y}-{\bf z})$, that is, ${\bf z}$ is adjacent to ${\bf x}$ and ${\bf y}$.

\item Let $({\bf x}-{\bf y})^{\tr}A({\bf x}-{\bf y})=a_1^2$. The proof is the same as in Case~1 with the only exception that here ${\bf w}:={\bf e}_1+{\bf e}_2$.\qedhere
\end{myenumerate}
\end{proof}

In determining the core of $VO_n^{\varepsilon}(q)$ we need to consider few of the elementary cases, with Witt index less than $2$, separately.
\begin{proposition}\label{majhni_primeri}
Let $\Gamma$ be one of the graphs $VO_2^{-}(q)$, $VO_2^{+}(q)$, or $VO_3(q)$. Then statements $(i)$, $(ii)$, and $(iii)$ from Proposition~\ref{klikaneod2} hold.
\end{proposition}
\begin{proof}
By Proposition~\ref{klikaneod2} it suffices to show $(iii)$. The graph $VO_2^{-}(q)$ is formed by $q^2$ isolated vertices, so $\omega\big(VO_2^{-}(q)\big)=1$, $\alpha\big(VO_2^{+}(q)\big)=q^2$, and $(iii)$ holds.

If $\Gamma=VO_2^{+}(q)$, then we may assume that the defining matrix $A$ in (\ref{e3}) equals
$$\left(\begin{array}{rr}1&0\\
0&-1\end{array}\right).$$
Then, $\{(0,x)^{\tr} : x\in \FF_q\}$ is an independent set. In fact,
$$
\big((0,x)^{\tr}-(0,y)^{\tr}\big)^{\tr}A\big((0,x)^{\tr}-(0,y)^{\tr}\big)=-(x-y)^2\neq 0
$$
for all distinct $x,y\in\FF_q$, so $\alpha(\Gamma)\geq q$. Since $|V(\Gamma)|=q^2$ and $\omega(\Gamma)=q$ by Corollary~\ref{lema2}, (\ref{e23}) implies $(iii)$.

In the case $\Gamma=VO_3(q)$ we may assume that
$$A=\left(\begin{array}{rrr}1&0&0\\
0&-1&0\\
0&0&d\end{array}\right),$$
where $d\in\FF_q$ is non-square. Then  $\{(0,x,y)^{\tr} : x,y\in \FF_q\}$ is an independent set, since
\begin{align*}
\big((0,x_1,y_1)^{\tr}-(0,x_2,y_2)^{\tr}\big)^{\tr}A\big((0,x_1,y_1)^{\tr}&-(0,x_2,y_2)^{\tr}\big)=
\\=&-(x_1-x_2)^2+d(y_1-y_2)^2\neq 0
\end{align*}
for $(0,x_1,y_1)^{\tr}\neq (0,x_2,y_2)^{\tr}$. So $\alpha(\Gamma)\geq q^2$. We proceed in the same way as in the case $\Gamma=VO_2^{+}(q)$.
\end{proof}

It is not so rare, that the core of a graph is complete. In fact, in~\cite[Theorem~4.1]{godsil_clanek} it was shown that the core of a connected regular graph, with the automorphism group acting transitively on pairs of vertices at distance two, is either complete or the graph itself is a core. Though the automorphism group of $VO_n^{\varepsilon}(q)$ does not have this property in general,
we will infer from Corollary~\ref{lema4} that a particular orbit of such an action is large enough to obtain a result for $VO_n^{\varepsilon}(q)$ that is analogous to~\cite[Theorem~4.1]{godsil_clanek}.
\begin{lemma}\label{lema5}
The graph $VO_n^{\varepsilon}(q)$ is either a core or its core is a complete graph on $\omega\big(VO_n^{\varepsilon}(q)\big)$ vertices.
\end{lemma}
\begin{proof}
For graph $VO_2^{-}(q)$ the result is proved in Proposition~\ref{majhni_primeri}, so we may assume that $VO_n^{\varepsilon}(q)$ is connected. Let $A$ be its defining matrix~(\ref{e3}). For ${\bf x}\neq{\bf y}$ we write ${\bf x}\sim {\bf y}$ if equation (\ref{e3}) is satisfied, that is, if vertices ${\bf x}$ and ${\bf y}$ are adjacent in $VO_n^{\varepsilon}(q)$.

Let $\Gamma'$ be a core of $VO_n^{\varepsilon}(q)$. Assume that $\Gamma'$ is neither complete nor the whole graph.
Since any endomorphism maps a clique of size $\omega\big(VO_n^{\varepsilon}(q)\big)$ to a clique of size $\omega\big(VO_n^{\varepsilon}(q)\big)$, $\Gamma'$ must contain such a clique ${\cal K}$. Since $VO_n^{\varepsilon}(q)$ is a connected graph, so is its core. Since $\Gamma'$ is not complete, we deduce that there is ${\bf u}\in V(\Gamma')\backslash {\cal K}$ and ${\bf v}\in {\cal K}$ such that ${\bf u}\sim {\bf v}$.

Now, since $VO_n^{\varepsilon}(q)$ is connected and $\Gamma'\neq VO_n^{\varepsilon}(q)$, there exists ${\bf w}_1$ outside $\Gamma'$ that is adjacent to some ${\bf v}_1$ in $\Gamma'$. Let $\Psi$ be any retraction of $VO_n^{\varepsilon}(q)$ onto $\Gamma'$. Then it maps ${\bf w}_1$ to some neighbor ${\bf u}_1$ of ${\bf v}_1$ in $\Gamma'$. By Lemma~\ref{prop1}, $VO_n^{\varepsilon}(q)$ is arc-transitive, so from Section~\ref{grafi} we know that $\Gamma'$ is arc-transitive as well. Consequently, there is an automorphism $\Phi'$ of $\Gamma'$ that maps the arc $({\bf u},{\bf v})$ to arc $({\bf u}_1,{\bf v}_1)$, so there exists a maximum clique ${\cal K}_1:=\Phi'({\cal K})$ in $\Gamma'$ that contains ${\bf v}_1$, while ${\bf u}_1\notin {\cal K}_1$. Define ${\bf x}_1:={\bf u}_1-{\bf v}_1$ and ${\bf y}_1:={\bf w}_1-{\bf v}_1$. Then ${\bf x}_1\sim 0\sim {\bf y}_1$ and ${\bf x}_1\nsim {\bf y}_1$, so ${\bf x}_1^{\tr}A{\bf x}_1=0={\bf y}_1^{\tr}A{\bf y}_1$ and ${\bf x}_1^{\tr}A{\bf y}_1\neq 0$.
Since the clique ${\cal K}_2:={\cal K}_1-{\bf v}_1$ is of maximum size and ${\bf x}_1={\bf u}_1-{\bf v}_1\notin {\cal K}_2$, there is ${\bf y}_2\in {\cal K}_2$ such that ${\bf y}_2\nsim {\bf x}_1$. Since $0={\bf v}_1-{\bf v}_1\in {\cal K}_2$, it follows that ${\bf y}_2^{\tr}A{\bf y}_2=({\bf y}_2-0)^{\tr}A({\bf y}_2-0)=0$, so we deduce that ${\bf x}_1^{\tr}A{\bf y}_2\neq 0$. By Corollary~\ref{lema4} there exists $P\in GL_n(\FF_q)$ and  nonzero $\alpha\in\FF_q$ such that
\begin{equation}\label{e43}
P^{\tr}AP=A,
\end{equation}
$P{\bf x}_1={\bf x}_1$, and $P{\bf y}_1=\alpha {\bf y}_2$. Since the clique ${\cal K}_2$ is of maximum size and contains the zero vector, it is a (totally isotropic) vector space. Since it contains ${\bf y}_2$, it follows that $\alpha {\bf y}_2\in {\cal K}_2$. Consequently, $\alpha {\bf y}_2+{\bf v}_1\in {\cal K}_1$. By (\ref{e43}), the map $\Phi({\bf x}):=P{\bf x}-(P{\bf v}_1-{\bf v}_1)$ is an automorphism of $VO_n^{\varepsilon}(q)$.
Moreover, it satisfies $\Phi({\bf u}_1)={\bf u}_1$ and $\Phi({\bf w}_1)=\alpha {\bf y}_2+{\bf v}_1$. Hence,
the restriction of the composition $\Phi'':=\Psi\circ\Phi^{-1}$ to $\Gamma'$ is an endomorphism of $\Gamma'$ that is not bijective, since $\Phi''({\bf u}_1)={\bf u}_1=\Phi''(\alpha {\bf y}_2+{\bf v}_1)$. This contradicts the fact that $\Gamma'$ is a core.
\end{proof}

\subsection{Spectrum}

As we shall see, the spectrum will provide us partial answer regarding the (non)completeness of the core of affine polar graphs.

Since the graphs  $VO_n^{-}(q)$ and $VO_n^{+}(q)$ are strongly regular with known parameters, their spectrum is easy to derive, well known, and mentioned already in the survey paper~\cite[p.~375, $C.12.^\pm$]{hubaut} (see also~\cite{www}).
\begin{lemma}\label{lastne_hyper_ell} (cf.~\cite{www,hubaut})
Let $n\geq 2$ be even. The eigenvalues $\lambda_i$ and their multiplicities $m_{\lambda_i}$ of the elliptic and hyperbolic affine polar graphs are as follows:
\begin{center}
\begin{tabular}{|l|l|l|}
\hline
 &$\lambda_1=(q^{\frac{n}{2}-1} - 1)(q^{\frac{n}{2}} + 1)$& $m_{\lambda_1}=1$\\
$VO_n^{-}(q)$ &$\lambda_2=q^{\frac{n}{2}-1}-1$ & $m_{\lambda_2}=q^{\frac{n}{2}-1}(q-1)(q^{\frac{n}{2}}+1)$\\
&$\lambda_3=- q^{\frac{n}{2}} + q^{\frac{n}{2}-1} - 1$ & $m_{\lambda_3}= (q^{\frac{n}{2}-1} - 1)(q^{\frac{n}{2}} + 1)$ \\
 \hline
 \hline
&$\lambda_1=(q^{\frac{n}{2}-1} + 1)(q^{\frac{n}{2}} - 1)$  &  $m_{\lambda_1}= 1$\\
$VO_n^{+}(q)$ & $\lambda_2=q^{\frac{n}{2}} - q^{\frac{n}{2}-1} - 1$ & $m_{\lambda_2}= (q^{\frac{n}{2}-1} + 1)(q^{\frac{n}{2}} - 1)$ \\
& $\lambda_3=-q^{\frac{n}{2}-1}-1$ & $m_{\lambda_3}= q^{\frac{n}{2}-1}(q-1)(q^{\frac{n}{2}}-1) $\\
 \hline
\end{tabular}.
\end{center}
\end{lemma}
\begin{remark}
Eigenvalues $\lambda_1$, $\lambda_2$, $\lambda_3$ are distinct, except in the case of a graph $VO_2^{-}(q)$, where $\lambda_1=0=\lambda_2$ and $m_{\lambda_3}=0$.
\end{remark}

\begin{remark}
Let ${\cal M}_{\Gamma}$ be the set of all real symmetric $|V(\Gamma)|\times |V(\Gamma)|$ matrices $M$, which have 1 at $(i,j)$-th entry, whenever $i=j$ or $i$-th and $j$-th vertex are not adjacent. The Lov\'{a}sz's $\vartheta$-function of a graph $\Gamma$ is defined as $\vartheta(\Gamma)=\inf_{M\in{\cal M}_{\Gamma}} \lambda_{\max} (M)$, where $\lambda_{\max} (M)$ is the largest eigenvalue of $M$~\cite{brouwerhaemers}. As proved in~\cite{lovasz}, for vertex-transitive graphs $\alpha(\Gamma)\leq \vartheta(\Gamma)\leq \frac{-|V(\Gamma)|\lambda_{\min}(\Gamma)}{\lambda_{\max}(\Gamma)-\lambda_{\min}(\Gamma)}$ and $\vartheta(\Gamma)\vartheta(\overline{\Gamma})=|V(\Gamma)|$ hold, where $\lambda_{\max}(\Gamma)$ and $\lambda_{\min}(\Gamma)$ are the largest and the smallest eigenvalue of $\Gamma$, respectively. If $\Gamma=\inv{VO_n^{+}(q)}$, then Corollary~\ref{lema2} implies that $\alpha(\Gamma)=\omega(\inv{\Gamma})=q^{\frac{n}{2}}$. Since eigenvalues of $\Gamma$ are easily computed from the eigenvalues of its complement (cf.~\cite[p.~4]{brouwerhaemers}), which are given in Lemma~\ref{lastne_hyper_ell}, we deduce that $\frac{-|V(\Gamma)|\lambda_{\min}(\Gamma)}{\lambda_{\max}(\Gamma)-\lambda_{\min}(\Gamma)}=q^{\frac{n}{2}}$. Therefore $\vartheta\big(\inv{VO_n^{+}(q)}\big)=q^{\frac{n}{2}}$ and consequently $\vartheta\big(VO_n^{+}(q)\big)=q^{\frac{n}{2}}$.
\end{remark}

Distinct eigenvalues of $VO_n(q)$ are described in the last column in the character table in~\cite[Case~3, p.~6129]{bannai} (the table was firstly computed in~\cite[Theorem~2]{kwok}, but it contained a misprint, as observed in~\cite{bannai}). These eigenvalues are $\lambda_1=q^{n-1}-1$, $\lambda_2=q^{\frac{n-1}{2}}-1$, $\lambda_3=-1$, and $\lambda_4=-q^{\frac{n-1}{2}}-1$. We were not able to find their multiplicities in the literature.
One possible strategy to compute them is to use the intersection matrices in~\cite[Section~3.7]{kwok} and apply the procedure described in~\cite[p.~46]{BCN} and~\cite[Proposition~2.2.2.a]{BCNcor}. However, since $q$ and $n$ are general (and not given fixed numbers), it seems that this tactic is too complicated. Therefore we apply a strategy from~\cite{JACO}, which will derive the multiplicities and recompute the eigenvalues $\lambda_1,\lambda_2,\lambda_3,\lambda_4$. The following result was essentially observed already in~\cite[Proof of Theorem~1.3]{JACO} though not written in a such generality.
\begin{proposition}\label{karakterji}
Let $V$ be a finite dimensional vector space over $\FF_q$.
\begin{enumerate}
\item  The $|V|$ characters of the group $(V,+)$ are precisely the maps
\begin{equation}\label{e24}
    \xi({\bf v}):=e^{2\pi i \Psi({\bf v})/p},
\end{equation}
    where $\Psi$ ranges over all $\FF_p$-linear maps $\Psi : V\to \FF_p$.

\item Let $\Gamma=Cay(V,S)$ be a Cayley graph, where $S$ is closed under multiplication by scalars in $\FF_p$.  Then the eigenvalues of $\Gamma$ are precisely the values
\begin{equation}\label{e26}
\frac{|\triangle| p-|S|}{p-1},
\end{equation}
where $\triangle:=\{{\bf s}\in S : \Psi({\bf s})=0\}$ and $\Psi$ ranges over all $\FF_p$-linear maps $\Psi : V\to \FF_p$.
\end{enumerate}
\end{proposition}
\begin{remark}
The statement and proof of Proposition~\ref{karakterji} is valid for $q$ odd or even.
\end{remark}
\begin{proof}
$(i)$ Obviously, the maps (\ref{e24}) are characters of $(V,+)$, that is, maps that satisfy $\xi({\bf u}+{\bf v})=\xi({\bf u})\xi({\bf v})$. Since the number of all characters and the number of all $\FF_p$-linear maps $\Psi : V\to \FF_p$ both equal $|V|=p^{\dim_p V}$ (cf.~\cite[Theorem~5.5]{finite-fields-LN}), where $\dim_p V$ is the dimension of $V$ as a vector space over $\FF_p$, it suffices to show that distinct $\Psi_1$ and $\Psi_2$ generate distinct characters. So assume that $e^{2\pi i \Psi_1({\bf v})/p}=e^{2\pi i \Psi_2({\bf v})/p}$ for all~${\bf v}$. Then there exist integers $k({\bf v})$ such that $2\pi i \Psi_1({\bf v})/p=2\pi i \Psi_2({\bf v})/p+2\pi  i k({\bf v})$, that is, $\Psi_1({\bf v})=\Psi_2({\bf v})+p k({\bf v})$. Hence, $\Psi_1=\Psi_2$ (mod $p$).

$(ii)$ By Lemma~\ref{karakterji-spekter} and $(i)$, the eigenvalues of $\Gamma$ are precisely the values
\begin{equation}\label{e25}
\sum_{{\bf s}\in S}e^{2\pi i \Psi({\bf s})/p},
\end{equation}
where $\Psi$ ranges over all $\FF_p$-linear maps $\Psi : V\to \FF_p$. For such map $\Psi$ let $\triangle:=\{{\bf s}\in S : \Psi({\bf s})=0\}$. Then, for ${\bf s}\in S\backslash \triangle$, we have $$\sum_{a\in \FF_p\backslash\{0\}}e^{2\pi i \Psi(a{\bf s})/p}=\sum_{a\in \FF_p\backslash\{0\}}e^{2\pi i a \Psi({\bf s})/p}=\sum_{j=1}^{p-1}e^{2\pi i j/p}=\sum_{j=0}^{p-1}e^{2\pi i j/p}-1=-1.$$
Consequently, (\ref{e25}) equals
$$\sum_{{\bf s}\in \triangle}e^{2\pi i \Psi({\bf s})/p}+\sum_{{\bf s}\in S\backslash \triangle}e^{2\pi i \Psi({\bf s})/p}=|\triangle|+\frac{|S\backslash \triangle|\cdot (-1)}{p-1}=|\triangle|+\frac{\big(|S|-|\triangle|\big)\cdot (-1)}{p-1},$$
which is the same as (\ref{e26}).
\end{proof}

\begin{corollary}\label{spekter-parabolic}
Let $n\geq 3$ be odd. The eigenvalues $\lambda_i$ and their multiplicities $m_{\lambda_i}$ of the parabolic affine polar graph are as follows:
\begin{center}
\begin{tabular}{|l|l|l|}
\hline
$VO_n(q)$ &$\lambda_1=q^{n-1}-1$& $m_{\lambda_1}=1$\\
 &$\lambda_2=q^{\frac{n-1}{2}}-1$& $m_{\lambda_2}=\frac{1}{2}(q-1)(q^{n-1}+q^{\frac{n-1}{2}})$\\
&$\lambda_3=-1$ & $m_{\lambda_3}=q^{n-1}-1$\\
&$\lambda_4=-q^{\frac{n-1}{2}}-1$ & $m_{\lambda_4}=\frac{1}{2}(q-1)(q^{n-1}-q^{\frac{n-1}{2}})$ \\
 \hline
\end{tabular}.
\end{center}
\end{corollary}

\begin{proof}
Graph $VO_n(q)$ is a Cayley graph for the additive group $(\FF_q^n,+)$ and the set
$S=\{{\bf x}\in\FF_q^n\backslash\{0\} : {\bf x}^{\tr}A{\bf x}=0\}$, where we may assume that $A=I$ is the identity matrix. The eigenvalues of  $VO_n(q)$ are, by Proposition~\ref{karakterji}, precisely the values (\ref{e26}),
where $\triangle:=\{{\bf s}\in S : \Psi({\bf s})=0\}$ and $\Psi$ ranges over all $\FF_p$-linear maps $\Psi : \FF_q^n\to \FF_p$. Any such map is of the form $\Psi({\bf x})=\sum_{j=1}^n \psi_j(x_j)$ for some $\FF_p$-linear maps $\psi_j : \FF_q\to\FF_p$, where ${\bf x}=(x_1,\ldots,x_n)^{\tr}$. By~(\ref{sled-lin}), there are scalars $a_{j}\in\FF_q$ such that $\psi_j(x_j)=\Tr(a_j x_j)$. Since the trace map is additive,
we deduce that
\begin{equation}\label{e28}
\Psi({\bf x})=\Tr ({\bf a}^{\tr}{\bf x}),
\end{equation}
where ${\bf a}=(a_1,\ldots,a_n)^{\tr}$. Obviously, the map (\ref{e28}) is $\FF_p$-linear for any ${\bf a}\in\FF_q^n$, and two distinct ${\bf a}_1,{\bf a}_2\in\FF_q^n$ generate two distinct maps.

To proceed, consider the map $x\mapsto x^p-x$ on $\FF_q$.
It is
$p$-to-$1$, since in a field of characteristic~$p$ the equivalence
$$x^p-x=y^p-y \Longleftrightarrow (x-y)^p=x-y \Longleftrightarrow
x-y\in \FF_p$$ holds. Consequently, the map $x\mapsto x^p-x$ attains $q/p$
distinct values $d_1,\ldots,d_{q/p}$, where one of them, say
$d_{q/p}$, is zero. Equivalence~(\ref{sled})
implies that
\begin{equation*}\label{4}
\Tr(x)=0\Longleftrightarrow x\in\{d_1,\ldots,d_{q/p}\},
\end{equation*}
so~(\ref{e28}) shows that
\begin{equation}\label{e29}
\triangle=\big\{{\bf x}\in \FF_q^n\backslash\{0\} : {\bf x}^{\tr}{\bf x}=0\ \textrm{and}\ {\bf a}^{\tr}{\bf x}\in\{d_1,\ldots,d_{q/p}\}\big\}.
\end{equation}
Given ${\bf a}\in\FF_q^n$ and $b\in\FF_q$ let $\Omega_b^{\bf a}:=\{{\bf x}\in \FF_q^n\backslash\{0\} : {\bf x}^{\tr}{\bf x}=0\ \textrm{and}\ {\bf a}^{\tr}{\bf x}=b\}$. If $b,c\in\FF_q$ are nonzero, then $|\Omega_b^{\bf a}|=|\Omega_c^{\bf a}|$, since ${\bf x}\mapsto \frac{c}{b}{\bf x}$ is a bijection between the two sets. By Lemma~\ref{lema1}, $|S|=q^{n-1}-1$.
In fact, $n$ is odd and we need to exclude the zero vector.
On the contrary, $S$ equals the disjoint union $\bigcup_{b\in \FF_q} \Omega_b^{\bf a}$, so $|S|=(q-1)| \Omega_1^{\bf a}|+ |\Omega_0^{\bf a}|$, that is, $|\Omega_1^{\bf a}|=\frac{|S|-|\Omega_0^{\bf a}|}{q-1}$. Therefore (\ref{e29}) implies that
\begin{align*}
|\triangle|&=\left|\bigcup_{j=1}^{q/p} \Omega_{d_j}^{\bf a}\right|
=(q/p-1)|\Omega_1^{\bf a}|+|\Omega_0^{\bf a}|
=\frac{q/p-1}{q-1} \big(|S|-|\Omega_0^{\bf a}|\big)+|\Omega_0^{\bf a}|\\
&= \frac{q-q/p}{q-1}|\Omega_0^{\bf a}|+\frac{q/p-1}{q-1}|S|,
\end{align*}
an the eigenvalue (\ref{e26}) equals
\begin{equation}\label{e30}
\frac{q |\Omega_0^{\bf a}|-|S|}{q-1}=\frac{q |\Omega_0^{\bf a}|-q^{n-1}+1}{q-1}.
\end{equation}
We need to compute $|\Omega_0^{\bf a}|$ for each ${\bf a}=(a_1,\ldots,a_n)^{\tr}\in\FF_q^n$. The second of the following two cases splits in three subcases, for a total of four distinct eigenvalues~(\ref{e30}).
\begin{myenumerate}{Case}
\item Let ${\bf a}=0$. Then $\Omega_0^{\bf a}=S$, and the eigenvalue (\ref{e30}) equals $\lambda_1:=q^{n-1}-1$.

\item\label{c1} Let Let ${\bf a}\neq 0$. Without lost of generality, we may assume that $a_1\neq 0$. Then
\begin{equation}\label{e33}
\Omega_0^{\bf a}=\left\{{\bf x}\in \FF_q^n\backslash\{0\} : x_1=-\frac{\sum_{j=2}^{n} a_jx_j}{a_1},\ (x_2,\ldots,x_n)A_1(x_2,\ldots,x_n)^{\tr}=0\right\},
\end{equation}
where $A_1$ is the matrix in (\ref{e27}) with determinant $\frac{{\bf a}^{\tr}{\bf a}}{a_1^2}$.\medskip

\emph{Case~2a}. Assume that $(-1)^{(n-1)/2}{\bf a}^{\tr}{\bf a}$ is a nonzero square in $\FF_q$. By Lemma~\ref{lema1},
$|\Omega_0^{\bf a}|=q^{n-2}+q^{\frac{n-1}{2}}-q^{\frac{n-3}{2}}-1$, so  the eigenvalue (\ref{e30}) equals $\lambda_2:=q^{\frac{n-1}{2}}-1$. Since half of nonzero elements in $\FF_q$ are squares, Lemma~\ref{lema1} implies that there are $m_{\lambda_2}:=\frac{1}{2}(q-1)(q^{n-1}+q^{\frac{n-1}{2}})$ vectors ${\bf a}\in\FF_q^n$ that satisfy the assumption of Case~2a.
\medskip

\emph{Case~2b}. Assume that $(-1)^{(n-1)/2}{\bf a}^{\tr}{\bf a}$ is non-square in $\FF_q$. By Lemma~\ref{lema1},
$|\Omega_0^{\bf a}|=q^{n-2}-q^{\frac{n-1}{2}}+q^{\frac{n-3}{2}}-1$, so  the eigenvalue (\ref{e30}) equals $\lambda_4:=-q^{\frac{n-1}{2}}-1$. Since half of nonzero elements in $\FF_q$ are non-squares, Lemma~\ref{lema1} implies that there are $m_{\lambda_4}:=\frac{1}{2}(q-1)(q^{n-1}-q^{\frac{n-1}{2}})$ vectors ${\bf a}\in\FF_q^n$ that satisfy the assumption of Case~2b.  \medskip

\emph{Case~2c}. Assume that ${\bf a}^{\tr}{\bf a}=0$. Since $a_1\neq 0$, there is $j\geq 2$ such that $a_{j}\neq 0$. Without lost of generality, we may assume that $a_2\neq 0$. Then $A_1$ is singular, but its lower-right $(n-2)\times (n-2)$ block $A_{12}$ is invertible, since, by Corollary~\ref{matrika}, its determinant equals
$$\det A_{12}=\frac{a_1^2+a_{3}^2+a_4^2+\ldots +a_n^2}{a_1^2}=\frac{{\bf a}^{\tr}{\bf a}-a_2^2}{a_1^2}=-\frac{a_2^2}{a_1^2}\neq 0.$$
By Lemma~\ref{lema33} there is $P\in GL_{n-1}(\FF_q)$ such that
$$A_1=P^{\tr}\left(\begin{array}{cc}
0&0\\
0&A_{12}
\end{array}\right)P.$$
Hence, (\ref{e33}) and the bijective transformation $(y_2,\cdots,y_{n})^{\tr}:=P(x_2,\cdots,x_{n})^{\tr}$
yield
\begin{align*}
|\Omega_0^{\bf a}|&=\left|\left\{(x_2,\ldots,x_n)\in \FF_q^{n-1} :  (x_2,\ldots,x_n)A_1(x_2,\ldots,x_n)^{\tr}=0\right\}\right|-1\\
&=q\cdot \left|\left\{(y_3,\ldots,y_n)\in \FF_q^{n-2} :  (y_3,\ldots,y_n)A_{12}(y_3,\ldots,y_n)^{\tr}=0\right\}\right|-1.
\end{align*}
Note that we have subtracted $1$ for the zero vector. By Lemma~\ref{lema1}, applied at $A_{12}$, we deduce that
$|\Omega_0^{\bf a}|=q^{n-2}-1$, so  the eigenvalue (\ref{e30}) equals $\lambda_3:=-1$. By Lemma~\ref{lema1} there are $m_{\lambda_3}:=q^{n-1}-1$ vectors ${\bf a}\in\FF_q^n$ that satisfy the assumption of Case~2c.\qedhere
\end{myenumerate}
\end{proof}

\begin{remark}
It follows from Corollary~\ref{spekter-parabolic} that $VO_n(q)$ is a \emph{Ramanujan graph} (if $q$ is odd), that is, a regular graph that satisfy $|\lambda|\leq 2\sqrt{\lambda_1-1}$ for all its eigenvalues $\lambda\neq \pm \lambda_1$, excluding the ones with the largest absolute value (this was observed already in~\cite[Theorem~3.3.(i)]{bannai}). Ramanujan graphs are good expanders and precisely those regular graphs, for which their Ihara zeta function satisfies an analog of Riemann hypothesis (cf.~\cite{murty}).
\end{remark}

\begin{remark}
It follows from Corollary~\ref{spekter-parabolic} that graph $VO_n(q)$ ($n\geq 3$) has 4 distinct eigenvalues. Therefore it follows from~\cite[Proposition~21.2]{biggs} that $VO_n(q)$ is not a distance-regular graph, since its diameter equals two by Proposition~\ref{diameter}. We refer to~\cite{BCN} for the definition of a distance-regular graph.
\end{remark}

\begin{remark}
A similar method, as the one used in the proof of Corollary~\ref{spekter-parabolic}, can be used to prove Lemma~\ref{lastne_hyper_ell}. Hence, by \cite[Lemma~10.2.1]{godsil_knjiga}, we can easily deduce (the well known fact) that graphs $VO_n^{+}(q)$ and $VO_n^{-}(q)$ are strongly regular, since they have precisely three distinct eigenvalues (unless we consider the graph $VO_2^{-}(q)$).
\end{remark}

The spectrum and Lemma~\ref{lema5} provide us first major result on the core of an affine polar graph.
\begin{theorem}\label{noncomplete-core-elliptic}
Let $n\geq 4$ be even and $q$ be odd. Then $VO_n^{-}(q)$ is a core. In particular, the core of $VO_n^{-}(q)$ and the core of the complement $\overline{VO_n^{-}(q)}$ is not a complete graph, and $\omega\big(VO_n^{-}(q)\big)\alpha\big(VO_n^{-}(q)\big)< \big|V\big(VO_n^{-}(q)\big)\big|$.
\end{theorem}
\begin{proof}
By Proposition~\ref{klikaneod2} and Lemma~\ref{lema5} it suffices to show that
\begin{equation}\label{e40}
\alpha\big(VO_n^{-}(q)\big)<q^{\frac{n}{2}+1},
\end{equation}
since  $\big|V\big(VO_n^{-}(q)\big)\big|=q^n$ and $\omega\big(VO_n^{-}(q)\big)=q^{\frac{n}{2}-1}$ by Corollary~\ref{lema2}. By Lemma~\ref{hoffman} and Lemma~\ref{lastne_hyper_ell}, we deduce the inequality $$\alpha\big(VO_n^{-}(q)\big)\leq q^n\cdot \frac{-(- q^{\frac{n}{2}} + q^{\frac{n}{2}-1} - 1)}{(q^{\frac{n}{2}-1} - 1)(q^{\frac{n}{2}} + 1)-(- q^{\frac{n}{2}} + q^{\frac{n}{2}-1} - 1)}$$
with the right side that is strictly smaller than $q^{\frac{n}{2}+1}$ if $n>2$, so (\ref{e40}) holds.
\end{proof}
The reader may observe that in the hyperbolic and parabolic case the proof above fails. In fact, we will see in the next subsection that an analog of Theorem~\ref{noncomplete-core-elliptic} for hyperbolic and parabolic affine polar graphs is false in some cases. There are however some cases, where such an analogous result is correct, but the spectrum alone seems to not provide enough information to prove it.

\subsection{Ovoids}

In this subsection we investigate the (non)completeness of the core of $VO_n^{\varepsilon}$ with certain geometrical tools. Therefore we assume the usual assumption from finite geometry that the Witt index $r$ in (\ref{wittindex}) is at least two, that is, $n\geq 5$, $n\geq 4$, and $n\geq 6$, if the affine polar graph is parabolic, hyperbolic, and elliptic, respectively.

An \emph{ovoid} of an orthogonal polar space is a subset of the vertex set $V\big(Q_{n-1}^{\varepsilon}(q)\big)$ meeting every generator in exactly one point. If an ovoid $\mathcal{O}$ exists, then its cardinality is well known (cf.~\cite{thas_hand}). In fact, it can be deduced by computing the cardinality of the set $\{({\bf x}, U) : {\bf x}\in \mathcal{O},\ U\ \textrm{is a generator that contains}\ {\bf x} \}$, which equals $|\mathcal{O}|\cdot g_0=g$, where $g$ is the number of all generators and $g_0$ is the number of all generators that contain a particular point. If $s$ is the number of points in a generator, then $g\cdot s=|V\big(Q_{n-1}^{\varepsilon}(q)\big)|\cdot g_0$, so
\begin{equation}\label{e17}
|\mathcal{O}|=\frac{g}{g_0}=\frac{|V\big(Q_{n-1}^{\varepsilon}(q)\big)|}{s}.
\end{equation}
Hence, the following lemma is deduced from (\ref{e16}).
\begin{lemma}(cf.~\cite{thas_hand})\label{ovoid}
Let $r\geq 2$. Assume that an ovoid $\mathcal{O}$ exists in $Q_{n-1}^{\varepsilon}(q)$. Then $|\mathcal{O}|$ equals  $q^{\frac{n-1}{2}}+1$, $q^{\frac{n}{2}-1}+1$, and $q^{\frac{n}{2}}+1$ if the orthogonal polar space is parabolic, hyperbolic, and elliptic, respectively.
\end{lemma}
A \emph{partial ovoid} $\mathcal{PO}$ is a subset of the vertex set $V\big(Q_{n-1}^{\varepsilon}(q)\big)$ meeting every generator in at most one point. If $g'$  and $g''$ are the number of generators that contain one and zero points of $\mathcal{PO}$, respectively, then $g=g'+g''$. Moreover, a similar calculation as in (\ref{e17}) shows that
$|\mathcal{PO}|=\frac{g'}{g_0}$. So if a partial ovoid has cardinality (\ref{e17}), then it is an ovoid, since $g''=0$.  Since $Q_{n-1}^{\varepsilon}(q)$ is arc-transitive (cf.~\cite{cameron_kan}), any pair of its adjacent vertices lie in some generator. Hence,  partial ovoids are precisely the independent sets in $Q_{n-1}^{\varepsilon}(q)$. Consequently, the following is true.
\begin{corollary}\label{posledica}
Let $r\geq 2$. If $\alpha\big(Q_{n-1}^{\varepsilon}(q)\big)$ is the same as the cardinality $|{\cal O}|$ from Lemma~\ref{ovoid}, then the corresponding independent set is an ovoid.
\end{corollary}
Most of the observations above about ovoids, partial ovoids, and independent sets were already used in the proof of~\cite[Theorem~3.5]{cameron_kan}.

The `No-Homomorphism Lemma'~\cite{albertson, hahn_tardif, klavzar} states that if there exists some homomorphism $\Phi : \Gamma'\to \Gamma$ between two graphs, where $\Gamma$ is vertex-transitive, then
\begin{equation}\label{e18}
\frac{\alpha(\Gamma')}{|V(\Gamma')|}\geq \frac{\alpha(\Gamma)}{|V(\Gamma)|}.
\end{equation}
In particular, (\ref{e18}) holds if $\Gamma'$ is a subgraph in $\Gamma$. In Proposition~\ref{lema_no_hom2} (i) we essentially rewrite the proof of~\cite[Lemma~3.3]{hahn_tardif} to deduce a strict inequality in~(\ref{e18}) if $\Gamma$
is an affine polar graph and $\Gamma'$ is a closed neighborhood of a vertex. Consequently, in part (ii) the sizes of independent sets in $VO_n^{\varepsilon}(q)$ provide us lower bounds for the size of the largest partial ovoid in orthogonal polar space.
\begin{proposition}\label{lema_no_hom2}\leavevmode
\begin{enumerate}
\item
Let $N_0$ be the closed neighborhood of any vertex in $VO_n^{\varepsilon}(q)$. Then
\begin{equation}\label{e21}
\frac{\alpha(N_0)}{|V(N_0)|}> \frac{\alpha(VO_n^{\varepsilon}(q))}{V\big(VO_n^{\varepsilon}(q)\big)}.
\end{equation}
\item Let $r\geq 2$. Then
\begin{align*}
\alpha\big(Q_{n-1}(q)\big)&>\frac{1}{q}\cdot \alpha(VO_n(q)),\\
\alpha\big(Q_{n-1}^{+}(q)\big)&>\frac{q^{\frac{n}{2}}+q-1}{q^{\frac{n}{2}+1}}\cdot \alpha(VO_n^{+}(q)),\\
\alpha\big(Q_{n-1}^{-}(q)\big)&>\frac{q^{\frac{n}{2}}-q+1}{q^{\frac{n}{2}+1}}\cdot \alpha(VO_n^{-}(q)).
\end{align*}
\end{enumerate}
\end{proposition}
\begin{proof}
(i) Since $VO_n^{\varepsilon}(q)$ is vertex-transitive, we may assume that $N_0$ is the closed neighborhood of the zero vertex, that is, the subgraph induced by the set $\{{\bf x}\in\FF_q^n : {\bf x}^{\tr} A {\bf x}=0\}$, where $A$ is the matrix that defines the polar space.

Let $\Gamma=VO_n^{\varepsilon}(q)$ and $\Gamma'=N_0$. Let $\mathcal{J}(\Gamma)$ denote the family of all independent sets in $\Gamma$ of size $\alpha(\Gamma)$. Since $\Gamma$ is vertex-transitive, any its vertex lies in the same number, say $m$, of members of  $\mathcal{J}(\Gamma)$. If we count the elements of the set
$\big\{({\bf x},\mathcal{I}) : {\bf x}\in V(\Gamma), \mathcal{I}\in \mathcal{J}(\Gamma)\ \textrm{contains}\  {\bf x} \big\}$
in two different ways, we deduce that
\begin{equation}\label{e19}
 |V(\Gamma)|\cdot m = |\mathcal{J}(\Gamma)|\cdot \alpha(\Gamma).
\end{equation}
Obviously, $|\mathcal{I}\cap V(\Gamma')|\leq \alpha(\Gamma')$ for all $\mathcal{I}\in \mathcal{J}(\Gamma)$. Fix one independent set $\mathcal{I}_1\in \mathcal{J}(\Gamma)$ and ${\bf x}_1\in \mathcal{I}_1$. Then $\mathcal{I}_1-{\bf x}_1:=\{{\bf x}-{\bf x}_1 : {\bf x}\in \mathcal{I}_1\}\in \mathcal{J}(\Gamma)$. Since it contains the zero vector, its other elements are not in $V(\Gamma')$, that is, $|(\mathcal{I}-{\bf x}_1)\cap V(\Gamma')|=1<\alpha(\Gamma')$.
Consequently,
\begin{equation}\label{e20}
\sum_{\mathcal{I}\in \mathcal{J}(\Gamma)} |\mathcal{I}\cap V(\Gamma')| < \alpha(\Gamma')\cdot |\mathcal{J}(\Gamma)|.
\end{equation}
Since $\sum_{\mathcal{I}\in \mathcal{J}(\Gamma)} |\mathcal{I}\cap V(\Gamma')|=|V(\Gamma')|\cdot m$, we deduce (\ref{e21}) from (\ref{e19}) and (\ref{e20}).

(ii) Let $N$ and $N_0$ be the neighborhood and the closed neighborhood of the zero vector in $VO_n^{\varepsilon}(q)$, respectively. Since the zero vector is adjacent to all vertices in $N$, we have $\alpha(N_0)=\alpha(N)$. By Lemma~\ref{lexi}, (\ref{e14}), and (i) it follows that
\begin{align*}
\alpha(Q_{n-1}^{\varepsilon}(q))&=\alpha(Q_{n-1}^{\varepsilon}(q))\cdot \alpha(K_{q-1})=\alpha(N)\\
&=\alpha(N_0)>\frac{|V(N_0)|}{V\big(VO_n^{\varepsilon}(q)\big)}\cdot\alpha(VO_n^{\varepsilon}(q)).
\end{align*}
Now, $V\big(VO_n^{\varepsilon}(q)\big)=q^n$, while to compute $|V(N_0)|$ we need to multiply the cardinality of the quadric by $|\FF_q\backslash\{0\}|=q-1$, and then add 1 for the zero vector. The result follows.
\end{proof}

In~\cite[Theorem~3.5]{cameron_kan} it was determined (in terms of existence of ovoids, spreads, and partitions into ovoids),  when is the core of $Q_{n-1}^{\varepsilon}(q)$ or its complement complete. Recall from Remark~\ref{opomba} that the result was not (necessary) symmetric for the graph and its complement. Moreover, in the proof it was observed that $$\omega\big(Q_{n-1}^{\varepsilon}(q)\big)\alpha\big(Q_{n-1}^{\varepsilon}(q)\big)=V\big(Q_{n-1}^{\varepsilon}(q)\big)$$ if and only if an ovoid exists in the polar space (provided that $r\geq 2$). For affine polar graph we obtain the following result.
\begin{theorem}\label{noncomplete-core-par-hyper}
Let $r\geq 2$ and let $VO_n^{\varepsilon}(q)$ be parabolic or hyperbolic. Then the statements (i)-(v) are equivalent:
\begin{enumerate}
\item $\omega\big(VO_n^{\varepsilon}(q)\big)\alpha\big(VO_n^{\varepsilon}(q)\big)= \big|V\big(VO_n^{\varepsilon}(q)\big)\big|$,
\item a core of $VO_n^{\varepsilon}(q)$ is a complete graph,
\item a core of $\inv{VO_n^{\varepsilon}(q)}$ is a complete graph,
\item $VO_n^{\varepsilon}(q)$ is not a core,
\item $\chi\big(VO_n^{\varepsilon}(q)\big)=\omega\big(VO_n^{\varepsilon}(q)\big)$.
\end{enumerate}
If any of statements (i)-(v) is true, then
\begin{enumerate}
\setcounter{enumi}{5}
\item the polar space of $Q_{n-1}^{\varepsilon}(q)$ has an ovoid.
\end{enumerate}
\end{theorem}

Recall that equivalence of statements $(i)-(iii)$ is proved in Proposition~\ref{klikaneod2}, since $VO_n^{\varepsilon}(q)$ is a Cayley graph over an abelian group. Equivalence between $(ii)$ and $(iv)$ is proved in Lemma~\ref{lema5}. We know the equivalence between $(ii)$ and $(v)$ from Section~\ref{grafi}. So  Theorem~\ref{noncomplete-core-par-hyper} demands just the proof of implication $(i)\Rightarrow (vi)$. We provide two different proofs.

\begin{bew}{1}
Assume that $(i)$ holds. Then $$\alpha(VO_n^{\varepsilon}(q))=\frac{q^n}{\omega\big(VO_n^{\varepsilon}(q)\big)},$$
where $\omega\big(VO_n^{\varepsilon}(q)\big)$ is computed in Corollary~$\ref{lema2}$.

In the parabolic case we deduce that
$\alpha\big(Q_{n-1}(q)\big)>q^{(n-1)/2}$ from Proposition~\ref{lema_no_hom2}.
Since $\alpha\big(Q_{n-1}(q)\big)$ is an integer, it follows that  $\alpha(Q_{n-1}(q))\geq q^{(n-1)/2}+1$. From (\ref{e23}) and (\ref{e1}) we deduce that $\alpha(Q_{n-1}(q))\leq q^{(n-1)/2}+1$, so $\alpha(Q_{n-1}(q))=q^{(n-1)/2}+1$. By Corollary~\ref{posledica}, the polar space has an ovoid.

In the hyperbolic case we similarly deduce that
$$\alpha\big(Q_{n-1}^{+}(q)\big)>q^{n/2-1}+1-\frac{1}{q},$$
so $\alpha\big(Q_{n-1}^{+}(q)\big)\geq q^{n/2-1}+1$. The same procedure as in the parabolic case shows that the polar space has an ovoid.\cqfd
\end{bew}

\begin{bew}{2}
Assume that $(i)$ holds. Then, both for parabolic and hyperbolic case, Corollary~\ref{spekter-parabolic} and Lemma~\ref{lastne_hyper_ell}, together with Corollary~\ref{lema2},  imply that equality is attained in (\ref{hoffman_neodvisna}). Let ${\cal K}$ and ${\cal I}$ be a clique of size $\omega\big(VO_n^{\varepsilon}(q)\big)$ and independent set of size $\alpha\big(VO_n^{\varepsilon}(q)\big)$, respectively. Choose arbitrary ${\bf k}\in {\cal K}$ and ${\bf i}\in {\cal I}$. Then ${\cal K}':={\cal K}-{\bf k}$ and ${\cal I}':={\cal I}-{\bf i}$ are clique and independent set of the same size as ${\cal K}$ and ${\cal I}$, respectively. Recall from the proof of Proposition~\ref{klikaneod2}, that $\{{\bf k}'+{\cal I}' : {\bf k}'\in {\cal K}'\}$ is a partition of the vertex set $V\big(VO_n^{\varepsilon}(q)\big)$ into independent sets of size $\alpha\big(VO_n^{\varepsilon}(q)\big)$. Since $0=0+0\in 0 + {\cal I}'$, we see that $0\notin {\bf k}'+{\cal I}'$ for arbitrary nonzero ${\bf k}'\in {\cal K}'$, so by Lemma~\ref{hoffman} there are precisely $-\lambda_{\min}$ elements ${\bf x}_1,\ldots, {\bf x}_{-\lambda_{\min}}\in {\bf k}'+{\cal I}'$, that are adjacent to 0. That is, if $A$ is the defining matrix (\ref{e3}), then
\begin{align*}
{\bf x}_j^{\tr}A{\bf x}_j&=0\qquad (j=1,\ldots, -\lambda_{\min}),\\
{\bf x}_j^{\tr}A{\bf x}_k&=-\frac{1}{2}({\bf x}_j-{\bf x}_j)^{\tr}A({\bf x}_j-{\bf x}_k)\neq 0 \qquad (j\neq k).
\end{align*}
Since, by Corollary~\ref{spekter-parabolic} and Lemma~\ref{lastne_hyper_ell}, $-\lambda_{\min}$ equals $q^{(n-1)/2}+1$ and $q^{n/2-1}+1$ in parabolic and hyperbolic case, respectively, we deduce from Corollary~\ref{posledica} that $\langle{\bf x}_1\rangle,\ldots,\langle{\bf x}_{-\lambda_{\min}}\rangle$ form an ovoid in $Q_{n-1}^{\varepsilon}(q)$.\cqfd
\end{bew}

Recall from Theorem~\ref{noncomplete-core-elliptic}  that $(i)-(v)$ do not hold in the elliptic case. Neither does $(vi)$ if $n\geq 6$ \cite{thas1981}.
However, in the elliptic case, the two proofs of implication $(i)\Rightarrow (vi)$ above fails, so the spectral tools used to prove Theorem~\ref{noncomplete-core-elliptic} are indeed needed.

It is our concern, whether the statement $(i)-(v)$ in Theorem~\ref{noncomplete-core-par-hyper} are true or not. Certainly, if $(vi)$ does not hold, then neither do the statements $(i)-(v)$. However, whether $(vi)$ is true or not is still not known in general, despite a huge amount of the research in this area in last few decades. We now briefly list the main results of what is known and refer to a recent survey~\cite{ovoid_povzetek} for more details.

The parabolic polar space of $Q_{n-1}(q)$ does not posses ovoids if $n\geq 9$~\cite{GunaMoor1997}. Ovoids exist if $n=5$ (cf.~\cite[3.4.1.(i)]{GQ}). If $n=7$, the (non)existence of ovoids in parabolic space seems to be still an open problem, though there are some partial results of a mixed type, that is, ovoids do not exist if $q>3$ is a prime (see~\cite[Theorem~3]{okeefe} and \cite[Corollaries~1,2]{ball2006}) and do exist if $q$ is a power of $3$ (cf.~\cite[Theorem~17]{thas1992}).

The hyperbolic polar space of $Q_{n-1}^{+}(q)$ is even more mysterious. Existence of ovoids is known if $n=4$ or $n=6$ (cf.~\cite{thas_hand}). If $n=8$, the existence is known for prime $q$~\cite{conway, moorhouse1993}, for $q=p^k$ with $k$ odd and $p\equiv 2\ (\textrm{mod}~3)$~\cite{kantor}, and in some other special cases (see \cite{ovoid_povzetek} and references therein). A non-existence of ovoids is proved if $q$ is a power of a prime $p$ that satisfies the inequality $p^{\frac{n}{2}-1}>\binom{n+p-2}{n-1}-\binom{n+p-4}{n-1}$ \cite{moorhouse1995}. This is the case, for example, if $n=10$ and $p=3$, or $n=12$ and $p\in\{5,7\}$. By \cite[Theorem~8]{thas1992}, nonexistence of ovoids in $Q_{n-1}^{+}(q)$ implies nonexistence of ovoids in $Q_{n+1}^{+}(q)$, so there are no ovoids in  $Q_{n-1}^{+}(q)$ if $n\geq 10$, $p=3$ or $n\geq 12$, $p\in\{5,7\}$. It is believed by many mathematicians that there are no ovoids in $Q_{9}^{+}(q)$  (and consequently in $Q_{n-1}^{+}(q)$, $n\geq 10$) for general $q$ (cf.~\cite{ovoid_povzetek}).

We do not know, whether statement $(vi)$ from Theorem~\ref{noncomplete-core-par-hyper} implies statements $(i)-(v)$ or not. Nevertheless, to check whether $(vi)$ or $(i)-(v)$ is true seems to be of a similar difficulty, as indicated by
the following result together with Theorem~\ref{noncomplete-core-par-hyper}.

\begin{proposition}\label{propThm2}\leavevmode
\begin{enumerate}
\item Let $n\geq 5$ be odd and assume that $Q_{n-1}(q)$ has an ovoid. Then the statements $(i)-(v)$ in Theorem~\ref{noncomplete-core-par-hyper} are true for $VO_{n-2}(q)$ and $VO_{n-1}^{+}(q)$.
\item Let $n\geq 4$ be even and assume that $Q_{n-1}^{+}(q)$ has an ovoid. Then the statements $(i)-(v)$ in Theorem~\ref{noncomplete-core-par-hyper} are true for $VO_{n-2}^{+}(q)$.
\end{enumerate}
\end{proposition}
Observe that the proof below shows how to construct a maximum independent set that satisfies statement $(i)$ in Theorem~\ref{noncomplete-core-par-hyper} from an ovoid, for affine polar graphs/orthogonal polar spaces in question.
\begin{proof}
We may assume that the defining matrix  $A$ (\ref{e3}) is
\begin{equation}\label{matrika-prop}
\left(
\begin{array}{ccccc}
  &  &  &  & 1 \\
  &  &  & 1 &  \\
  &  & \iddots &  &  \\
  & 1 &  &  &  \\
 1 &  &  &  &
\end{array}
\right)
\end{equation}
in the parabolic case and
$$\left(
\begin{array}{cr}
 B & 0 \\
 0 & -1
\end{array}
\right),$$
in the hyperbolic case, where $B$ is $(n-1)\times (n-1)$ matrix of the form (\ref{matrika-prop}).
Then, for ${\bf x}=(x_1,\ldots,x_n)^{\tr}$, ${\bf x}^{\tr}A{\bf x}$ equals
\begin{equation}\label{e49}
x_{\frac{n+1}{2}}^2+2\sum_{i=1}^{\frac{n-1}{2}}x_i x_{n-i}
\end{equation}
in parabolic case and
\begin{equation}\label{e50}
x_{\frac{n}{2}}^2-x_n^2 +2\sum_{i=1}^{\frac{n-1}{2}}x_i x_{n-1-i}.
\end{equation}
in hyperbolic case.
Assume that ${\cal O}=\{\langle {\bf x}_1 \rangle,\ldots, \langle {\bf x}_t \rangle, \langle {\bf x}_{t+1} \rangle\}$ is an ovoid in $Q_{n-1}^{\varepsilon}(q)$, where $t$ equals $q^{\frac{n-1}{2}}$  and $q^{\frac{n}{2}-1}$ in the parabolic and hyperbolic case, respectively, as it follows from Lemma~\ref{ovoid}. Then
\begin{align}
\label{e47} {\bf x}_{i}^{\tr}A {\bf x}_{i}&=0\qquad (i=1,\ldots,t+1),\\
\label{e48} {\bf x}_{i_1}^{\tr}A {\bf x}_{i_2}&\neq 0\qquad (i_1\neq i_2),
\end{align}
since an ovoid is also an independent set in graph $Q_{n-1}^{\varepsilon}(q)$.
Obviously,  there is $i_0\in\{1,\ldots,t+1\}$ and $j_0\leq n$ such that $j_0$-th component of ${\bf x}_{i_0}$ is nonzero and $j_0\neq \frac{n+1}{2}$ in parabolic case and $j_0\notin\{n,\frac{n}{2}\}$ in hyperbolic case. To simplify writings we can assume that $i_0=t+1$ and $j_0=1$.
Let $a_i$ be the first component of ${\bf x}_{i}$, so $a_{t+1}\neq 0$.
Let nonzero scalars $b_1,\ldots, b_t$ be such that
\begin{equation}\label{e46}
b_i{\bf x}_i^{\tr}A{\bf x}_{t+1}=1\qquad (i\leq t),
\end{equation}
and denote ${\bf y}_i:=b_i{\bf x}_i-\frac{a_i b_i}{a_{t+1}}{\bf x}_{t+1}$. Then
${\cal I}=\{{\bf y}_1,\ldots, {\bf y}_t\}$ is an independent set in $VO_{n}^{\varepsilon}(q)$, since (\ref{e47})-(\ref{e46}) imply that
\begin{align*}
&({\bf y}_{i_1}-{\bf y}_{i_2})^{\tr}A({\bf y}_{i_1}-{\bf y}_{i_2})=\\
&=\Big(b_{i_1}{\bf x}_{i_1}-b_{i_2}{\bf x}_{i_2}+\tfrac{a_{i_1} b_{i_1}-a_{i_2} b_{i_2}}{a_{t+1}}{\bf x}_{t+1}\Big)^{\tr}A\Big(b_{i_1}{\bf x}_{i_1}-b_{i_2}{\bf x}_{i_2}+\tfrac{a_{i_1} b_{i_1}-a_{i_2} b_{i_2}}{a_{t+1}}{\bf x}_{t+1}\Big)\\
&=(b_{i_1}{\bf x}_{i_1}-b_{i_2}{\bf x}_{i_2})^{\tr}A(b_{i_1}{\bf x}_{i_1}-b_{i_2}{\bf x}_{i_2})\\
&+2\tfrac{a_{i_1} b_{i_1}-a_{i_2} b_{i_2}}{a_{t+1}}\cdot (b_{i_1}{\bf x}_{i_1}-b_{i_2}{\bf x}_{i_2})^{\tr}A{\bf x}_{t+1}+\left(\tfrac{a_{i_1} b_{i_1}-a_{i_2} b_{i_2}}{a_{t+1}}\right)^2{\bf x}_{t+1}^{\tr}A{\bf x}_{t+1}\\
&=-2b_{i_1}b_{i_2}\cdot {\bf x}_{i_1}^{\tr}A{\bf x}_{i_2} + 0 +0 \neq 0
\end{align*}
for $i_1\neq i_2$. We next define vectors ${\bf z}_1,\ldots, {\bf z}_t\in\FF_q^{n-2}$  and a $(n-2)\times (n-2)$ matrix $A'$ as follows. In the parabolic case, we obtain ${\bf z}_i$ by deleting the first and the last component of ${\bf y}_i$. Similarly, $A'$ is obtained from $A$ by deleting the first/last row/column. In the hyperbolic case we delete the first and $(n-1)$-th component of ${\bf y}_i$ and first/$(n-1)$-th row/column of $A$ to obtain ${\bf z}_i$ and $A'$. Since the first component of ${\bf y}_i$ is zero by the construction, it follows from the forms (\ref{e49}) and (\ref{e50}) that ${\bf y}_{i_1}^{\tr}A{\bf y}_{i_2}={\bf z}_{i_1}^{\tr}A'{\bf z}_{i_2}$ for all $i_1$ and $i_2$, so
${\cal I}'=\{{\bf z}_1,\ldots, {\bf z}_t\}$ is an independent set of size $t$ in graph $VO_{n-2}^{\varepsilon}(q)$, which is defined by matrix $A'$. So for this affine polar graph the statement $(i)$ in Theorem~\ref{noncomplete-core-par-hyper} is satisfied. This proves the claim $(ii)$ and half of claim $(i)$ from Proposition~\ref{propThm2}. Now, if ${\cal I}'$ is an independent set of size $t=q^{\frac{n-1}{2}}$ in $VO_{n-2}(q)$, then, by extending its members with zero value as the last component, we obtain an independent set $\{({\bf z}^{\tr},0)^{\tr} : {\bf z}\in {\cal I}'\}$ of size $q^{\frac{n-1}{2}}$ in $VO_{n-1}^{+}(q)$, where the defining matrix equals
$$\left(
\begin{array}{cr}
 A' & 0 \\
 0 & -1
\end{array}
\right).$$
This completes the proof of claim $(i)$.
\end{proof}

As it was mentioned above, it is of great interest in finite geometry to know whether $Q_6(q)$ and $Q_{n-1}^{+}(q)$, $n\geq 8$ have ovoids for particular values of $q$. If one wants to show that there are no ovoids, then, by Proposition~\ref{propThm2}, it is sufficient to show that the statement $(i)$ in Theorem~\ref{noncomplete-core-par-hyper} is not true for $VO_5(q)$ or $VO_6^{+}(q)$ and $VO_{n-2}^{+}(q)$, respectively. For some values of $q$ this might be easier, since there are less dimensions to consider.

Recall that there are ovoids in $Q_{4}(q)$, $q$ arbitrary;  $Q_{6}(3^k)$; $Q_{7}^{+}(q)$, $q$ prime or $q=p^k$ with $k$ odd and $p\equiv 2\ (\textrm{mod}~3)$. Hence, by Proposition~\ref{propThm2}, the following graphs satisfy statements $(i)-(v)$ in Theorem~\ref{noncomplete-core-par-hyper}:
\begin{align*}
&VO_{4}^{+}(q),\ q\ \textrm{odd},\\
&VO_{5}(3^k),\ k\ \textrm{arbitrary},\\
&VO_{6}^{+}(3^k),\ k\ \textrm{arbitrary},\\
&VO_{6}^{+}(q),\ q\ \textrm{odd prime},\\
&VO_{6}^{+}(q),\ q=p^k,\ p\ \textrm{and}\ k\ \textrm{odd},\ p\equiv 2\ (\textrm{mod}~3).
\end{align*}
In some cases of affine polar graphs of Witt index $<2$,  independent sets that attains the equality in statement $(i)$ in Theorem~\ref{noncomplete-core-par-hyper}, i.e. those with cardinality
\begin{equation}\label{maxneodvisne}
\frac{\big|V\big(VO_n^{\varepsilon}(q)\big)\big|}{\omega\big(VO_n^{\varepsilon}(q)\big)},
\end{equation}
were already constructed in the proof of Proposition~\ref{majhni_primeri}. Below we present few examples with Witt index $\geq 2$.

\begin{example}\label{primer0}
Consider the graph $VO_{4}^{+}(q)$, where $A=\diag(1,-d,-d,1)$ is the defining matrix (\ref{e3}) and $d\in\FF_q$ is a non-square. Then ${\cal I}=\Big\{(x,y,0,0)^{\tr} : x,y,\in\FF_q\Big\}$ is an independent set. In fact, if
${\bf u}_i=(x_i,y_i,0,0)^{\tr}$ for $i=1,2$, then
$$({\bf u}_1-{\bf u}_2)^{\tr}A({\bf u}_1-{\bf u}_2)=(x_1-x_2)^2-d(y_1-y_2)^2\neq 0$$ for ${\bf u}_1\neq {\bf u}_2$.
The size of ${\cal I}$ equals $q^2$, i.e., (\ref{maxneodvisne}).
\end{example}

Independent sets from Example~\ref{primer1} and  Example~\ref{primer2} are constructed by applying the technique described in the proof of Proposition~\ref{propThm2} to Thas-Kantor ovoids of $Q_{6}(3^k)$. The model used to represent the ovoid is from~\cite{williams}.
\begin{example}\label{primer1}
Consider the graph $VO_5(q)$, where $q=3^k$ for some $k\geq 1$, and the defining matrix (\ref{e3}) equals
\begin{equation}\label{5x5}
A=\left(\begin{array}{rrrrr}
 0 & 0 & 0 & 0 & -1 \\
 0 & 0 & 0 & -1 & 0 \\
 0 & 0 & 1 & 0 & 0 \\
 0 & -1 & 0 & 0 & 0 \\
 -1 & 0 & 0 & 0 & 0
\end{array}
\right).
\end{equation}
Let $d\in\FF_q$ be a non-square and let
$${\cal I}=\Big\{(x,y,z,x^2y-dy^3-xz,-d^{-1}x^3+xy^2+yz)^{\tr} : x,y,z\in\FF_q\Big\}.$$
If  ${\bf u}_i=(x_i,y_i,z_i,x_i^2y_i-dy_i^3-x_iz_i,-d^{-1}x_i^3+x_iy_i^2+y_iz_i)^{\tr}$ for $i=1,2$, then a direct computation with a use of equality $3=0$ shows that
\begin{align*}
({\bf u}_1&-{\bf u}_2)^{\tr}A({\bf u}_1-{\bf u}_2)=\\
=&-d^{-1}\Big((x_1-x_2)^2-d(y_1-y_2)^2\Big)^2+\Big((z_1-z_2)+(x_2y_1-x_1y_2)\Big)^2,
\end{align*}
which equals zero only if ${\bf u}_1={\bf u}_2$, so ${\cal I}$ is an independent set of size (\ref{maxneodvisne}).
\end{example}

\begin{example}\label{primer2}
Consider the graph $VO_6^{+}(q^k)$, where $q=3^k$ for some $k\geq 1$, and the defining matrix equals
\begin{equation}\label{primer2eq}
\left(
\begin{array}{cr}
 A & 0 \\
 0 & -1
\end{array}
\right),
\end{equation}
where $A$ is the matrix in (\ref{5x5}). Let ${\cal I}'\subseteq \FF_q^6$ be formed from the set ${\cal I}$ in Example~\ref{primer1} by extending its members with $0$ in the last entry, that is, ${\cal I}'=\{({\bf x}^{\tr},0)^{\tr} : {\bf x}\in {\cal I}\}$. Then ${\cal I}'$ is an independent set of size (\ref{maxneodvisne}).
\end{example}

\begin{example}\label{primer3}
Let $-3$ be a non-square in $\FF_q$, or equivalently  $q=p^k$, $p$ and $k$ both odd, with $p\equiv 2\ (\textrm{mod}~3)$ (the equivalency follows from the Law of quadratic reciprocity (cf.~\cite[Theorem~5.17]{finite-fields-LN}) and some elementary finite field theory). Then there are ovoids in $Q_{7}^{+}(q)$~\cite{kantor} and we can use the technique from the proof of Proposition~\ref{propThm2} to construct an independent set in $VO_{6}^{+}(q)$ of size (\ref{maxneodvisne}). We treat only the case when $-1$ is a non-square and $3$ is a square (the construction, when  $-1$ is a square and $3$ is a non-square, is similar).

Let
\begin{equation}\label{6x6p2mod3}
A=\left(\begin{array}{rrrrrr}
 -3 & 0 &  &  &  &  \\
 0 & 1 &  &  &  &  \\
  &  & 0 & -1 &  &  \\
  &  & -1 & 0 &  &  \\
  &  &  &  & 0 & 1 \\
  &  &  &  & 1 & 0
\end{array}\right)
\end{equation}
be the defining matrix and let
$${\cal I}=\Big\{(x, y, xz+yw, z, yz-xw, w)^{\tr} : y,z,w\in\FF_q, x=-(z^2+ w^2)/2\Big\}.$$
If  ${\bf u}_i=(x_i, y_i, x_iz_i+y_iw_i, z_i, y_iz_i-x_iw_i, w_i)^{\tr}$ for $i=1,2$, then by evaluating both sides of the equation
$$({\bf u}_1-{\bf u}_2)^{\tr}A({\bf u}_1-{\bf u}_2)=\Big(x_1+x_2+z_1z_2+w_1w_2\Big)^2+\Big(y_1-y_2+z_2w_1-z_1w_2\Big)^2,$$
with a use of equality $x_i=-(z_i^2+ w_i^2)/2$, we see that they are equal indeed. By Lemma~\ref{lema1}, the sum of two squares is zero only if both squares are zero, so we see that the expression above is zero only if ${\bf u}_1={\bf u}_2$. Hence, ${\cal I}$ is an independent set of size (\ref{maxneodvisne}).
\end{example}

Bijective maps which preserve cones on nonsingular non-anisotropic metric spaces of dimension $\geq 3$, with a symmetric bilinear form, were characterized in~\cite{lester1977}. In particular, this result characterizes the automorphisms of affine polar graph $VO_n^{\varepsilon}(q)$ for $q$ odd and $n\geq 3$. Hence, for the elliptic case this theorem characterizes also all graph endomorphisms, as it follows from Theorem~\ref{noncomplete-core-elliptic}. The same holds for parabolic and hyperbolic case if the polar space related to $Q_{n-1}^{\varepsilon}(q)$ has no ovoids, as it follows from Theorem~\ref{noncomplete-core-par-hyper}. This characterization even simplifies a bit if $q$ is a prime or if the defining matrix $A$ in (\ref{e3}) has entries that are fixed by all automorphisms of the field $\FF_q$ (for example, if entries are just $0$, $1$, $-1$), as we shall see in the next section for the graph obtained from Minkowski space.

\section{Finite Minkowski spaces}\label{sectionMinkowski}

In what follows we assume that $-1\in\FF_q$  is not a square, which is equivalent to $q\equiv 3$ $(\textrm{mod}~4)$ (cf.~\cite[p.~135]{wan2}). Then, by~(\ref{e41}), the field is of the form
\begin{equation}\label{e45}
\FF_q=\{0,x_1^2,\ldots, x_{\frac{q-1}{2}}^2, -x_1^2,\ldots, -x_{\frac{q-1}{2}}^2\}.
\end{equation}
In the literature devoted to some alternative theories in particle physics that are based on geometry over finite fields (cf.~\cite{ahmavaara1,ahmavaara2,ahmavaara3,ahmavaara4,coish,shapiro,belbla1968JMP,belbla1968NC,blasi}),
squares $x_1^2,\ldots, x_{\frac{q-1}{2}}^2$ and non-squares $-x_1^2,\ldots, -x_{\frac{q-1}{2}}^2$ are interpreted as `positive' and `negative' numbers, respectively. Though the product of two `positive' or two `negative' numbers is a `positive' number and the product of a `positive' and a `negative' number is `negative', this analogy with real numbers breaks down with addition, since a sum of two squares is not necessarily a square. However, it was shown in~\cite{kustaanheimo1950} that, for certain prime values $q$, the elements $1,2,3,\ldots, N$ are all squares for certain $N$, so the sum of two such elements is a square if it does not exceed $N$. Consequently, a relation $$a>b\myeq a-b\ \textrm{is `positive'}$$  becomes transitive for a large subset of consecutive numbers. There are infinite number of such prime values $q$ and value $N$ can be made arbitrary large by selecting sufficiently large $q$ (see~\cite{kustaanheimo1950,coish,ahmavaara1} for more details). By comparing the difference between two consecutive numbers in $1,2,3,\ldots, N$ with the `smallest observable length' one can consider an approximation between real and finite field coordinates, and consequently an approximation between large subsets of mathematical structures over a large finite field and corresponding structures over real numbers (cf.~\cite{ahmavaara1}). It is not the purpose of this section to take a `philosophical  point of view', which would be in favor of either finite or real geometry to be used as a model for some physical theory. Our research provides a better understanding of the mathematical background in finite case, and shows some of the fundamental geometric differences between finite and real Minkowski space that go beyond the approximation of coordinates.

A \emph{finite Minkowski space} of dimension $n$ is the vector space $\FF_q^n$ equipped with the inner product
\begin{equation}\label{e4}
({\bf x},{\bf y}):=x_1y_1-x_2y_2-\cdots -x_n y_n,
\end{equation}
where ${\bf x}:=(x_1,\ldots,x_n)^{\tr}$ and  ${\bf y}:=(y_1,\ldots,y_n)^{\tr}$. Two such \emph{events} ${\bf x}$ and ${\bf y}$ are \emph{light-like} if $({\bf x}-{\bf y}, {\bf x}-{\bf y})=0$. A map $\Phi: \FF_q^n\to \FF_q^n$ that maps light-like events into light-like events, that is, $\big(\Phi({\bf x})-\Phi({\bf y}),\Phi({\bf x})-\Phi({\bf y})\big)=0$ whenever $({\bf x}-{\bf y}, {\bf x}-{\bf y})=0$, is said to \emph{preserve the speed of light}.
Observe that the product (\ref{e4}) can be simply written as $({\bf x},{\bf y})={\bf x}^{\tr}M{\bf y}$ for the diagonal matrix $M:=\diag(1,-1,\ldots,-1)\in SGL_n(\FF_q)$. We say that matrices $L$ and $K$, of size $n\times n$, are \emph{Lorentz} and \emph{anti-Lorentz}, if
\begin{equation}\label{e11}
(L{\bf x},L{\bf y})=({\bf x},{\bf y})\quad \textrm{and} \quad (K{\bf x},K{\bf y})=-({\bf x},{\bf y})
\end{equation}
holds for all ${\bf x},{\bf y}\in\FF_q^n$, respectively. The two conditions in~(\ref{e11}) are equivalent to
$L^{\tr}ML=M$ and $K^{\tr}MK=-M$, respectively. By applying the determinant to the last equation we see that anti-Lorentz matrices do not exist for odd $n$, since the opposite would imply that $(\det K)^2=(-1)^n=-1$, which contradicts our assumption that $-1$ is not a square. The existence of anti-Lorentz matrices for 4-dimensional Minkowski space is well known (cf.~\cite{coish}) and essentially the same construction works for other even dimensions. Namely, by Lemma~\ref{lema1} there are $a_0,b_0\in\FF_q$ such that $a_0^2+b_0^2=-1$, so the block diagonal matrix
$$\left(\begin{array}{rrrrrrr}
0&1&  & &      &  & \\
1&0&  & &      &  & \\
 & & a_0&b_0&      &  & \\
 & &-b_0&a_0&      &  & \\
 & &  & &\ddots&  & \\
 & &  & &      &a_0 &b_0\\
 & &  & &      &-b_0&a_0
\end{array}\right),$$
is anti-Lorentz. This differs from the real case, where the Sylvester's law of inertia forbids the existence of anti-Lorentz matrices for $n\geq 3$.

Maps $\Phi$ on finite Minkowski space of dimension $n\geq 3$ that are bijective and such that both $\Phi$ and the inverse $\Phi^{-1}$ preserve the speed of light, were classified in~\cite{lester1977}. Actually, this result considered more general spaces. It was proved that any such map is of the form $\Phi({\bf x})={\cal L}({\bf x})+\Phi(0)$, where bijective map ${\cal L} :\FF_q^n\to\FF_q^n$ satisfy $\big({\cal L}({\bf x}),{\cal L}({\bf y})\big)=d_0 \tau \big(({\bf x},{\bf y})\big)$, ${\cal L}({\bf x}+{\bf y})={\cal L}({\bf x})+{\cal L}({\bf y})$, and ${\cal L}(a{\bf x})=\tau(a){\cal L}({\bf x})$ for all ${\bf x},{\bf y}\in \FF_q^n$ and $a\in\FF_q$. Here, $d_0\in\FF_q$ is a fixed nonzero scalar and $\tau :\FF_q\to\FF_q$ is a field automorphism. Now, given a field automorphism $\sigma$ and a column vector ${\bf x}$ or a matrix $A$, let ${\bf x}^{\sigma}$ and $A^{\sigma}$ denote the vector/matrix that is obtained by applying $\sigma$ entry-wise. Then, the map ${\bf x}\mapsto {\cal L}({\bf x})^{\tau^{-1}}$ is linear, i.e., a multiplication by some invertible matrix $P$, so ${\cal L}({\bf x})=(P{\bf x})^{\tau}$ for all ${\bf x}$. Since the entries of $M$ are just $1$, $-1$, and $0$, which are fixed by field automorphisms, we see that $M^{\tau}=M$. Consequently,
\begin{align*}
({\bf x}^{\tau})^{\tr}(P^{\tau})^{\tr}MP^{\tau}{\bf y}^{\tau}&= \big({\cal L}({\bf x}),{\cal L}({\bf y})\big) = d_0 \tau \big(({\bf x},{\bf y})\big)\\
&= d_0 ({\bf x}^{\tau})^{\tr}M^{\tau}{\bf y}^{\tau} = ({\bf x}^{\tau})^{\tr} (d_0 M){\bf y}^{\tau}
\end{align*}
for all ${\bf x}$ and ${\bf y}$, so $(P^{\tau})^{\tr}MP^{\tau}=d_0 M$. By (\ref{e45}), $d_0=\pm x_0^2$ for some nonzero $x_0$, so the matrix $x_0^{-1}P^{\tau}$ is either Lorentz or anti-Lorentz, and $\Phi({\bf x})=x_0 (x_0^{-1}P^{\tau}){\bf x}^{\tau}+\Phi(0)$. Hence the main theorem in~\cite{lester1977} simplifies into the following lemma in the case of a finite Minkowski space.

\begin{lemma}(cf.~\cite{lester1977})\label{lester_auto}
Let $n\geq 3$. A bijective map $\Phi: \FF_{q}^n\to \FF_q^n$ satisfies the rule
\begin{equation*}\label{e6}
({\bf x}-{\bf y}, {\bf x}-{\bf y})=0 \Longleftrightarrow  \big(\Phi({\bf x})-\Phi({\bf y}),\Phi({\bf x})-\Phi({\bf y})\big)=0
\end{equation*}
if and only if it is of the form
\begin{align}
\label{e7} \Phi({\bf x})&=a L {\bf x}^{\tau}+{\bf x}_0\\
\nonumber \textrm{or}&\\
\label{e7a} \Phi({\bf x})&=a K {\bf x}^{\tau}+{\bf x}_0\qquad (n\ \textrm{is even}),
\end{align}
where $0\neq a\in\FF_{q}$ and ${\bf x}_0\in\FF_q^n$ are fixed, $\tau$ is a field automorphism of $\FF_q$ that is applied entry-wise to ${\bf x}$,  while $L$ and~$K$ are Lorentz and anti-Lorentz matrices respectively.
\end{lemma}
\begin{remark}
Recall that in the case $q$ is a prime, the identity map is the unique automorphism $\tau$. More generally, if $q=p^k$, where $p$ is a prime, then the field automorphisms are precisely the maps $a\mapsto a^{p^j}$ for $0\leq j\leq k-1$ (cf.~\cite[Theorem~2.21]{finite-fields-LN}).
\end{remark}

We define a \emph{finite Minkowski graph} $M_n(q)$ as the affine polar graph with $M$ as the defining matrix $A$ in (\ref{e3}). Since $\det M=(-1)^{n-1}$, we deduce from Lemma~\ref{lema1} that
$$M_n(q)=\left\{
\begin{array}{ll}VO_n(q)\ &\textrm{if}\ n\ \textrm{is odd},\\
VO_n^{+}(q)\ &\textrm{if}\ n\equiv 2\, (\textrm{mod}\, 4),\\
VO_n^{-}(q)\ &\textrm{if}\ n\equiv 0\, (\textrm{mod}\, 4).
\end{array}
\right.$$
Observe that maps that are characterized in Lemma~\ref{lester_auto} are precisely the automorphisms of $M_n(q)$. We are now able to state and prove the main theorem of this paper.
\begin{theorem}\label{glavni}
Let $q \equiv 3\, (\textrm{mod}~4)$ and $n\geq 4$. Assume that one of the following conditions is satisfied:
\begin{align*}
& n\equiv 0\, (\textrm{mod}~4),\\
& n\geq 9\  \textrm{is odd},\\
& n=7\ \textrm{and}\ q>3\ \textrm{is a prime},\\
& n=7\ \textrm{and}\ Q_{n-1}(q)\ \textrm{does not have an ovoid},\\
& n\equiv 2\ (\textrm{mod}~4)\ \textrm{and}\ Q_{n-1}^{+}(q)\ \textrm{does not have an ovoid}.
\end{align*}
A map $\Phi: \FF_{q}^n\to \FF_q^n$ satisfies the rule
\begin{equation}\label{rule}
({\bf x}-{\bf y}, {\bf x}-{\bf y})=0, {\bf x}\neq {\bf y} \Longrightarrow  \big(\Phi({\bf x})-\Phi({\bf y}),\Phi({\bf x})-\Phi({\bf y})\big)=0, \Phi({\bf x})\neq \Phi({\bf y})
\end{equation}
if and only if it is of the form~(\ref{e7}) or~(\ref{e7a}), where the form~(\ref{e7a}) is possible only for even $n$.

For general $n\geq 4$, if there is a map $\Phi$ that satisfies the rule (\ref{rule}) and is neither of the form (\ref{e7}) nor (\ref{e7a}), then there exists a map that satisfies the rule (\ref{rule}) and has $\omega\big(M_n(q)\big)$ pairwise light-like events as its image.
\end{theorem}

\begin{remark}
For $n=4$ a different proof of Theorem~\ref{glavni} was already presented by the author in the arXiv version of~\cite{drugi_del}. For many dimensions, which include $n=4$, the finite and the real Minkowski space differ in an essential way, since in the real case there always exist non-bijective maps $\Phi : \RR^n\to\RR^n$ that satisfy the condition (\ref{rule}). In fact, since the field $\RR$ is infinite, there is an injection $f : \RR^n\to\RR$, so the map $\Phi({\bf x})=\big(f({\bf x}),f({\bf x}),0,\ldots,0\big)^{\tr}$ satisfies (\ref{rule}), since its image consists of pairwise light-like events.
\end{remark}
\begin{bew}{of Theorem~\ref{glavni}}
Maps that satisfy the rule (\ref{lester_auto}) are precisely the endomorphisms of graph $M_n(q)$. If $n\equiv 0\, (\textrm{mod}~4)$, then $M_n(q)=VO_n^{-}(q)$ is a core by Theorem~\ref{noncomplete-core-elliptic}, so any its endomorphism is an automorphism, and hence characterized in Lemma~\ref{lester_auto}. If $n$ is odd or $n\equiv 2\, (\textrm{mod}~4)$, then $M_n(q)$ equals $VO_n(q)$ or $VO_n^{+}(q)$, respectively. If the corresponding orthogonal polar space does not have an ovoid, then $M_n(q)$ is a core by Theorem~\ref{noncomplete-core-par-hyper} and the result follows as in the elliptic case. Recall that there are no ovoids in $Q_{n-1}(q)$ if $n\geq 9$ or if $n=7$ and $q>3$ is a prime.

For general $n\geq 4$, if there is a map obeying the rule (\ref{rule}) and not of the form (\ref{e7})-(\ref{e7a}), then it is a graph endomorphisms but not an automorphism, so, by Lemma~\ref{lema5}, the core of $M_n(q)$ is a complete graph on $\omega\big(M_n(q)\big)$ vertices. Consequently there exists an endomorphism with this complete graph  as its image.\cqfd
\end{bew}

Recall from Proposition~\ref{klikaneod2} that in the case that
\begin{equation*}\label{minkowski-neodvisna}
\alpha\big(M_n(q)\big)=\frac{V\big(M_n(q)\big)}{\omega\big(M_n(q)\big)}
\end{equation*}
holds, then there is a map that satisfies the rule (\ref{rule}) and has image formed by $\omega\big(M_n(q)\big)$ pairwise light-like events. In fact, if ${\cal K}$ and ${\cal I}$ are maximum clique and maximum independent set, respectively, then the map (\ref{coloring}) has this property. This map is of the form $\Phi({\bf x})={\bf x}_1$, where ${\bf x}_1\in {\cal K}$ and ${\bf x}_2\in {\cal I}$ are the unique vectors from these sets that satisfy ${\bf x}={\bf x}_1+{\bf x}_2$. Below we construct few examples of such maps. The cases $n=2$ and $n=3$ are trivial, while other maps are constructed from the independent sets in Examples~\ref{primer1},~\ref{primer2},~\ref{primer3}.

\begin{example}
If $n=2$, we can choose ${\cal K}=\{(x_1,x_1)^{\tr} : x_1\in\FF_q\}$ and ${\cal I}=\{(0,x_2)^{\tr} : x_2\in\FF_q\}$ as the clique and independent set in  $M_2(q)$, respectively. The decomposition ${\bf x}=(x_1,x_2)^{\tr}=(x_1,x_1)^{\tr}+(0,x_2-x_1)^{\tr}\in {\cal K}+{\cal I}$ implies that the map $\Phi({\bf x})=(x_1,x_1)^{\tr}$ satisfies the rule (\ref{rule}).

If $n=3$, we can choose ${\cal K}=\{(x_1,x_1,0)^{\tr} : x_1\in\FF_q\}$ and ${\cal I}=\{(0,x_2,x_3)^{\tr} : x_2,x_3\in\FF_q\}$, which induce the map $\Phi({\bf x})=(x_1,x_1,0)^{\tr}$, where ${\bf x}=(x_1,x_2,x_3)^{\tr}$.
\end{example}
\begin{example}\label{exa5}
Let $n=5$ and $q=3^k$, where $k$ is odd, so that $-1$ is not a square. Then $P^{\tr}MP=A$, where $A$ is the matrix (\ref{5x5}) and
\begin{equation}\label{pinverzp}
P=\left(
\begin{array}{rrrrr}
 1 & 0 & 0 & 0 & 1 \\
 0 & 1 & 1 & 1 & 0 \\
 0 & 1 & 0 & -1 & 0 \\
 0 & 1 & -1 & 1 & 0 \\
 1 & 0 & 0 & 0 & -1
\end{array}
\right)\quad \textrm{with}\quad P^{-1}=\left(
\begin{array}{rrrrr}
 -1 & 0 & 0 & 0 & -1 \\
 0 & 1 & -1 & 1 & 0 \\
 0 & -1 & 0 & 1 & 0 \\
 0 & 1 & 1 & 1 & 0 \\
 -1 & 0 & 0 & 0 & 1
\end{array}
\right).
\end{equation}
It follows from Example~\ref{primer1} that
$$P\Big\{\big(x_1,x_2,x_3,f(x_1,x_2,x_3),g(x_1,x_2,x_3)\big)^{\tr} : x_1,x_2,x_3\in\FF_q\Big\},$$
with $f(x_1,x_2,x_3)=x_1^2x_2+x_2^3-x_1x_3$ and $g(x_1,x_2,x_3)=x_1^3+x_1x_2^2+x_2x_3$, is an independent set in $M_5(q)$. Similarly, $P\{(0,0,0,x_4,x_5)^{\tr} : x_4,x_5\in\FF_q\}$ is a clique. From the decomposition
\begin{align*}
P{\bf x}=P(x_1,x_2,x_3,x_4,x_5)^{\tr}&=P(0,0,0,x_4-f(x_1,x_2,x_3),x_5-g(x_1,x_2,x_3))^{\tr}\\
&+P(x_1,x_2,x_3,f(x_1,x_2,x_3),g(x_1,x_2,x_3))^{\tr},
\end{align*}
we deduce the map $\Phi(P{\bf x})=P(0,0,0,x_4-f(x_1,x_2,x_3),x_5-g(x_1,x_2,x_3))^{\tr}$. If $(y_1,y_2,y_3,y_4,y_5)^{\tr}:=P{\bf x}$, that is, $${\bf x}=P^{-1}{\bf y}=(-y_1-y_5,y_2-y_3+y_4,-y_2+y_4,y_2+y_3+y_4,-y_1+y_5)^{\tr},$$ then we obtain map
$$\Phi({\bf y})=P\left(\begin{array}{c}
 0\\
 0\\
 0\\
 y_2+y_3+y_4-f(-y_1-y_5,y_2-y_3+y_4,-y_2+y_4)\\
 -y_1+y_5-g(-y_1-y_5,y_2-y_3+y_4,-y_2+y_4)
\end{array}\right)$$
that satisfies the rule (\ref{rule}) in explicit form.
\end{example}
\begin{example}
Let $n=6$ and $q=3^k$, where $k$ is odd. Then the matrix (\ref{primer2eq}) equals $Q^{\tr}MQ$, where
$$Q=\left(
\begin{array}{cc}
 P & 0 \\
 0 & 1
\end{array}
\right)$$
and $P$ is the matrix in (\ref{pinverzp}).
It follows from Example~\ref{primer2} that
$$Q\Big\{\big(x_1,x_2,x_3,f(x_1,x_2,x_3),g(x_1,x_2,x_3),0\big)^{\tr} : x_1,x_2,x_3\in\FF_q\Big\},$$
where $f, g$ are defined in Example~\ref{exa5}, is an independent set in $M_6(q)$. The set $Q\{(0,0,x_3,x_4,x_5,x_3)^{\tr} : x_3,x_4,x_5\in\FF_q\}$ is a clique. We can use the decomposition
\begin{align*}
Q{\bf x}&=Q(0,0,x_6,x_4-f(x_1,x_2,x_3-x_6),x_5-g(x_1,x_2,x_3-x_6),x_6)^{\tr}\\
&+Q(x_1,x_2,x_3-x_6,f(x_1,x_2,x_3-x_6),g(x_1,x_2,x_3-x_6),0)^{\tr},
\end{align*}
and proceed in the same way as in Example~\ref{exa5}, to obtain map
$$\Phi({\bf y})=Q\left(\begin{array}{c}
 0\\
 0\\
 y_6\\
 y_2+y_3+y_4-f(-y_1-y_5,y_2-y_3+y_4,-y_2+y_4-y_6)\\
 -y_1+y_5-g(-y_1-y_5,y_2-y_3+y_4,-y_2+y_4-y_6)\\
 y_6
\end{array}\right)$$
that satisfies the rule (\ref{rule}).
\end{example}
\begin{example}
Let $n=6$ and $q=p^k$, where $k$ is odd and $p\equiv 11$ $(\textrm{mod}~12)$. Note that, by Chinese remainder theorem and Law of quadratic reciprocity, this is just a compact way of saying that $-1$ is non-square, $3$ is a square, $k$ is odd, and $p\equiv 2$ $(\textrm{mod}~3)$, so that the constructions in Example~\ref{primer3} are valid.  Let $a_0, b_0, c_0$ be such that $a_0^2+b_0^2=-1$ and $c_0^2=3$.
Then $P^{\tr}MP=A$, where $A$ is the matrix (\ref{6x6p2mod3}) and
\begin{equation*}
P=\left(
\begin{array}{cccccc}
 0 & 1 & 0 & 0 & 0 & 0 \\
 0 & 0 & 0 & 0 & \frac{1}{2} & -1 \\
 0 & 0 & \frac{1}{2} & 1 & 0 & 0 \\
 0 & 0 & \frac{a_0}{2} & -a_0 & \frac{b_0}{2} & b_0 \\
 0 & 0 & \frac{b_0}{2} & -b_0 & -\frac{a_0}{2} & -a_0 \\
 c_0 & 0 & 0 & 0 & 0 & 0
\end{array}
\right)
\end{equation*}
with
\begin{equation*}
\quad P^{-1}=\left(
\begin{array}{cccccc}
 0 & 0 & 0 & 0 & 0 & \frac{1}{c_0} \\
 1 & 0 & 0 & 0 & 0 & 0 \\
 0 & 0 & 1 & -a_0 & -b_0 & 0 \\
 0 & 0 & \frac{1}{2} & \frac{a_0}{2} & \frac{b_0}{2} & 0 \\
 0 & 1 & 0 & -b_0 & a_0 & 0 \\
 0 & -\frac{1}{2} & 0 & -\frac{b_0}{2} & \frac{a_0}{2} & 0
\end{array}
\right).
\end{equation*}
It follows from Example~\ref{primer3} that
$$P\Big\{\big(f(x_4,x_6), x_2, g(x_2,x_4,x_6), x_4, h(x_2,x_4,x_6), x_6\big)^{\tr} : x_2,x_4,x_6\in\FF_q\Big\},$$
where
\begin{align*}
f(x_4,x_6)&=-(x_4^2+x_6^2)/2,\\
g(x_2,x_4,x_6)&=f(x_4,x_6)x_4+x_2x_6,\\
h(x_2,x_4,x_6)&=x_2x_4-f(x_4,x_6)x_6,
\end{align*}
is an independent set in $M_6(q)$. Similarly, $P\{(x_1, c_0 x_1,x_3,0,x_5,0)^{\tr} : x_1, x_3, x_5\in\FF_q\}$ is a clique. We can use the decomposition
\begin{align*}
P{\bf x}&=P\left(\begin{array}{c}x_1-f(x_4,x_6)\\
c_0x_1-c_0f(x_4,x_6)\\
x_3 - g\big(x_2-c_0x_1+c_0f(x_4,x_6),x_4,x_6\big)\\
0\\
x_5- h\big(x_2-c_0x_1+c_0f(x_4,x_6),x_4,x_6\big)\\
0
\end{array}\right)\\
&+P\left(\begin{array}{c}f(x_4,x_6)\\
x_2-c_0x_1+c_0f(x_4,x_6)\\
g\big(x_2-c_0x_1+c_0f(x_4,x_6),x_4,x_6\big)\\
x_4\\
h\big(x_2-c_0x_1+c_0f(x_4,x_6),x_4,x_6\big)\\
x_6
\end{array}\right),
\end{align*}
and proceed in the same way as above, to obtain map
$$\Phi({\bf y})=P\left(\begin{array}{c}\frac{y_6}{c_0}-f\big(s,t\big)\\
y_6-c_0f\big(s,t\big)\\
y_3-a_0y_4-b_0y_5- g\Big(y_1-y_6+c_0 f\big(s,t\big),s,t\Big)\\
0\\
y_2-b_0y_4+a_0y_5- h\Big(y_1-y_6+c_0 f\big(s,t\big),s,t\Big)\\
0
\end{array}\right)$$
that satisfies the rule (\ref{rule}). Here,
\begin{align*}
s=s(y_3,y_4,y_5)&:=\frac{y_3+a_0y_4+b_0y_5}{2},\\
t=t(y_2,y_4,y_5)&:=\frac{-y_2-b_0y_4+a_0y_5}{2}.
\end{align*}
\end{example}

Recall that there exist a map
$\Phi : \FF_q^6\to\FF_q^6$ that satisfy the rule (\ref{rule}) and has the image formed by pairwise light-like events,  whenever $q$ is a prime of the form  $q\equiv 3$ $(\textrm{mod}~4)$. In fact, the image of such a map is a 3 dimensional (totally isotropic) vector subspace or its translation. Unfortunately, for arbitrary $q$ of this form, we were not able to construct such a map explicitly, since the Conway et al/Moorhouse ovoids in $Q_{7}^{+}(q)$~\cite{conway, moorhouse1993}, which guarantee the existence of such a map, are not parameterized.\bigskip

\noindent{\bf Acknowledgements.}
The author is thankful
to Andries E. Brouwer for pointing him to reference~\cite{hubaut}. He would also like to thank Gabriel Verret for pointing him to reference~\cite{roberson}, which subsequently led to the observation that Proposition~\ref{klikaneod2} has already been proved in~\cite{godsil_notes}.

\end{document}